\documentclass[11pt]{article}

\usepackage[top=0.9in,bottom=1.1in,left=1.05in,right=1.05in]{geometry}

\usepackage{amsmath,amssymb}
\usepackage{amsthm}
\usepackage{mathtools}
\usepackage{latexsym}
\usepackage{dsfont}
\usepackage{graphicx}
\usepackage{float}
\usepackage{color}
\usepackage{authblk}

\usepackage{verbatim}
\usepackage{epstopdf}

\usepackage{url}

\usepackage[utf8]{inputenc}

\newcommand\E{\ensuremath{\mathbb{E}}}
\newcommand\PP{\ensuremath{\mathbb{P}}}

\newcommand{\e}{\mathrm{e}}

\DeclareMathOperator\argmax{arg max}

\newtheorem{prop}{Proposition}[section]
\newtheorem{lemma}{Lemma}[section]
\newtheorem{defn}{Definition}[section]

\newtheorem{remark}{Remark}[section]

\title{Mean Field Game Approach to Production and Exploration of Exhaustible Commodities}
\author[1]{Michael Ludkovski\thanks{Ludkovski is partially supported by AMPS-1736439: ludkovski@pstat.ucsb.edu}}
\author[1]{Xuwei Yang}
\affil[1]{Dept of Statistics and Applied Probability\\ University of California Santa Barbara 93106-3110 USA}

\date{}

\numberwithin{equation}{section}

\begin{document}

\maketitle

\begin{abstract}
In a game theoretic framework, we study energy markets with a continuum of homogenous producers who produce energy from an exhaustible resource such as oil. Each producer simultaneously optimizes production rate that drives her revenues, as well as exploration effort to replenish her reserves.
This exploration activity is modeled through a controlled point process that leads to stochastic increments to reserves level.
The producers interact with each other through the market price that depends on the aggregate production.
We employ a mean field game approach to solve for a Markov Nash equilibrium and develop numerical schemes to solve the resulting system  of non-local HJB and transport equations with non-local coupling.
A time-stationary formulation is also explored, as well as the fluid limit where exploration becomes deterministic.

\bigskip


\end{abstract}


\section{Introduction}

We consider a stochastic differential game model for an exhaustible commodity market, such as oil. The dynamic market evolution is driven by the use of existing reserves to produce energy and exploration/discovery of new reserves. We assume a Cournot-type competition where each producer chooses their production rate; this resembles e.g.~OPEC members who adjust their output to influence crude oil prices. Extraction of the commodity generates a revenue stream but carries the depletion trade-off. To offset the resulting lower reserves, producers undertake efforts to explore for new reserves. Exploration is uncertain: continuous exploration efforts stochastically lead to discrete discoveries of additional reserves. Individually, producers aim to maximize their total expected profits which are equal to price times quantity extracted, minus the production and exploration costs. Strategically, the producers interact via the global market price $p$ that is determined by the aggregate production.

To model this oligopoly among commodity producers (i.e.~energy firms), we assume a continuum of homogenous agents who compete in a single market for energy. Each agent is small enough to be a price taker, yet in equilibrium the aggregate behavior fully determines supply and in turn the clearing price. For simplicity, we assume constant demand, focusing on producers' choices. This model is reasonable to describe the long-term behavior of the market (on the scale of years), where micro-economic fluctuations are averaged out and commodity supply and reserves is the main determinant of market structure.

While the literature on single-agent optimization for exhaustible resource extraction dates back to the 1970's \cite{DasguptaHealBook,Pindyck78,Pindyck80}, the first rigorous treatment of a dynamic continuous-time non-cooperative model  was by Harris et al.~\cite{HHS} in 2010. They studied an $N-$player Cournot game via the associated system of nonlinear HJB partial differential equations but due to numerical challenges illustrations were limited to the two-player model. Further analysis of competitive duopoly was carried out in \cite{LS-Cournot} and our earlier work \cite{LudkovskiYang14}; the special case of  a single exhaustible player competing against $N$ renewable producers of varying profitability was studied in \cite{SircarLedvina12,DasarathySircar14}.

\subsection{Mean field game approach}

In a differential game model with a finite number $N$ of players, their equilibrium strategies can be determined by a system of Hamilton-Jacobi-Bellman-Isaacs (HJB-I) equations derived from the dynamic programming principle. The dimension of the system in general increases in $N$, which makes the game model intractable for large $N$. Mean field game (MFG) approach simplifies the modeling by considering equilibria with a continuum of homogenous players; the respective finite dimensional game state expands into a distribution $m(\cdot)$. The main idea is then to consider the optimization problem of the representative agent; the latter becomes a regular stochastic control problem with the competitive effect captured via a mean-field interaction driven by $m(\cdot)$. In turn, the aggregate behavior of the players implies dynamics on the distribution $m$ of agent states. This leads to a system of two partial differential equations (PDE's) which is viewed as an approximation to a single multi-dimensional HJB-I PDE in the original finite-$N$ setup. We refer to \cite{BensoussanFrehseBook,CarmonaBook} for the general theory of MFGs.

In our context, the individual states are the reserves' levels $X$, and the interaction is via the market price $p$ that is related to total production $Q$ across all the producers.  Thus, $p$ enters the game value function of the representative producer as the mean field term and drives her choice of production and exploration controls. In turn, the distribution of reserves is driven by the latter production rates and exploration efforts. The key aspect of such oligopolistic MFG's, first introduced in Gu\'{e}ant's PhD thesis~\cite{GueantPhD,GLL10}, is the mean-field interaction via the aggregate $Q$. Since producers directly optimize their own production rates, this corresponds to a mean-field game of \emph{controls}, in contrast to the standard situation where the interaction is through the density $m(\cdot)$. A second distinguishing feature of Cournot MFG's is the hard exhaustibility constraint when reserves reach zero. Different possibilities  at $X_t = 0$ include leaving the game (which magnifies market power of remaining producers); switching to a renewable/inexhaustible resource; prolonging production via costly reserve replenishment. All the above choices lead to non-standard boundary conditions in the respective equations, necessitating tailored treatment.

Two further crucial aspects of oligopoly MFG concern the prescribed dynamics of the reserves process $(X_t)$ and the inverse demand curve that relates $p$ to aggregate production $Q$. For the latter aspect, starting with \cite{GLL10}, price was assumed to be linear (decreasing) in quantity, which brings forward some of the tractability of the original linear-quadratic MFG's. This choice was maintained in \cite{ChanSircar14,Graber16,GraberBensoussan15}. Very recently, Chan and Sircar \cite{ChanSircar16} also investigated MFG's with power-type demand curves. We note that even with linear price schedule, the overall link between production and price is non-linear due to the exhaustibility condition at $x=0$ which requires to separately keep track of exhausted and producing firms.

For an exhaustible resource, reserves are non-increasing and are completely determined by past production. However, this does not capture the real-life aspect of replenishment of fossil-fuel commodities, where global reserve depletion has to a large degree been offset with ongoing discoveries (deep-sea oil, shale gas, oil sands,  and so on). Such discoveries take place thanks to exploration activities determined by the respective exploration efforts. In the early model of Pindyck \cite{Pindyck78}, exploration was incremental, leading to deterministic reserve additions. Subsequent extensions represented exploration as a point process counting new reserve discoveries, see \cite{ArrowChang82,DeshmukhPliska80,HaganCaflisch94} which is also the choice we pursue below.
Another alternative, which is motivated by \emph{uncertainty} about current reserves, has been to introduce exogenous Brownian noise, i.e.~make reserves stochastically fluctuating. This is also convenient for theoretical and numerical purposes and has been commonly used in recent MFG literature, see \cite{ChanSircar14,GraberBensoussan15,GraberMouzouni17,GLL10}. Let us also mention further possibilities of $(X_t)$ following  a stochastic differential equation with controlled volatility and drift \cite{Pindyck80} and a common noise MFG model to capture systematic shocks to aggregate reserves~\cite{Graber16}. 
 
Mathematically, the MFG setup leads to an HJB equation to model a representative agent's strategy, and a transport equation to model the evolution of the distribution of all the producers' states. The structure of these equations is determined by the prescribed dynamics of $(X_t)$. When $X_t$ is deterministic, the equations are first-order, see e.g.~\cite{CardaliaguetGraber15}. When $X_t$ includes Brownian shocks (cf.~\cite{GLL10,ChanSircar14}), the HJB equation is second-order and the transport equation is the usual Kolmogorov forward equation. In contra-distinction, discrete discoveries
add a first-order non-local (``delay'') term to both the HJB \cite{LS-Cournot,LudkovskiYang14} and transport equations. These features are central to the numerical resolution of MFG's that requires handling a coupled system of nonlinear PDE's. We refer to \cite{Achdou13,Achdou16,CarliniSilva14} for a general summary of different computational approaches and their convergence, including finite difference and semi-Langrangian schemes. An alternative common approach
 \cite{GLL10,ChanSircar14,ChanSircar16} is based on Picard-like iterative schemes.

\subsection{Summary of Contributions}
In this paper we apply the MFG approach to model energy markets with a large population of competing producers of exhaustible but replenishable resources. Our main focus is the strategic interaction between exploration and production (E\&P), in a dynamic, stochastic, game-theoretic framework. E\&P is a major theme in the business decisions of energy firms, but is rarely tackled as such in mathematical models. Some of the topics we investigate are: (i) the price effect of exploration; (ii) aggregate production path implied by the model; (iii) aggregate exploration efforts; (iv) possibility of a stationary equilibrium where exploration exactly offsets production; (v) impact of exploration uncertainty on the solution. Our analysis yields quantitative insights into the macro behavior of commodity industries on long-time horizons, linking up with colloquial topics of ``peak oil'', ``value of exploration R\&D'' and ``postponing the exhaustion Doomsday''. Specifically, the stationary model where individual resource levels stochastically change, but the overall reserves distribution and associated aggregate production and price remain constant, is a feature that to our knowledge is new in the oligopoly MFG literature.

Relative to existing Cournot MFG models, we emphasize the analysis of stochastic exploration which leads to first-order, non-local MFG equations. In that sense, our work fits into two different strands of game-theoretic models of energy production. On the one hand, we extend  \cite{ChanSircar14,GLL10} who considered exhaustible resource MFGs but without exploration; thus reserves were non-increasing. On the other hand, we extend the duopoly model \cite{LS-Cournot} to the limiting oligopoly with a continuum of producers. In the duopoly each producer directly influences the price;  in the MFG model herein, each producer has negligible power on market price that is rather driven by the \emph{aggregate} production.

The closest work to ours is by Chan and Sircar \cite{ChanSircar16} who primarily focused on competition of exhaustible producers who switch to a costly renewable resource upon ultimate depletion. They also studied competition of a large group of exhaustible producers alongside a single renewable producer, similar to the major-minor model of Huang \cite{Huang10}. Section 5 of \cite{ChanSircar16} then briefly discusses resource exploration  and the respective stationary equilibrium. Compared to their illustration, we provide multiple additional analyses, including a detailed treatment of both the time-dependent and time-stationary equilibria, convergence to stationarity as the problem horizon increases, and study of the ``fluid limit''. The latter is a law-of-large-numbers scaling that maintains exploration/production controls but removes the associated uncertainty. This mechanism  allows to quantify the pure impact of uncertainty on the MFG model, linking to the deterministic first-order MFG, which is another new development relative to existing literature.

Our setup generates several numerical challenges due to the non-local terms in the PDE's and a non-local coupling between them; a major part of the paper is devoted to constructing a computational scheme to solve the MFG equations. Specifically, we decouple the HJB and transport equations via a Picard-like iteration that alternately updates the optimal production and exploration controls, and the reserves distribution function (which in turn determines the market price). For the HJB equation and similar to \cite{ChanSircar16}, we employ a method of lines, discretizing the space dimension and solving the resulting  ordinary differential equations in the time dimension. The latter still constitutes a coupled system of ODE's due to the exploration control and the implicit boundary condition at $x=0$. For the transport equation we use a fully explicit finite-difference scheme.  However, due to the non-smooth dynamics of $(X_t)$, rather than working with the density $m(dx)$, we operate with the corresponding cumulative distribution function $\eta(\cdot)$, and moreover separately treat the proportion $\pi(t)$ of exhausted producers.

This paper is organized as follows.
In Section~\ref{sec: Model}, we introduce the $N$-player Cournot game that motivates the MFG model in the limit $N\rightarrow \infty$.
Section~\ref{sec: Mean field game Nash equilibrium} discusses the doubly coupled system of HJB and transport equations that characterize the MFG Nash equilibrium. Section~\ref{sec: Numerical methods and examples} is devoted to numerical methods for solving this system and presents numerical illustrations. The rest of the paper then presents two modifications of the main model that yield important economic insights.
In Section~\ref{sec: Stationary mean field game Nash equilibrium}, we study the stationary MFG model, in which
the reserves distribution remains invariant due to the counteracting effects of production and exploration.
Section~\ref{sec: Fluid limit of exploration process} investigates the asymptotic ``fluid limit''  regime whereby the exploration process becomes deterministic, so that discovery of new resources
happens at high frequency with infinitesimal discovery amounts.
The paper concludes with Section \ref{sec:conclusion} and an Appendix that contains most of the proofs.


\section{Model}
\label{sec: Model}

\subsection{Finite player Cournot game }
\label{sec: Cournot game of $N$ players}

We consider an energy market with $N$ producers (players).
Each producer uses exhaustible resources, such as oil, to produce energy.
Let $X^i_t$ represent the reserves level of player $i$, $i=1,\ldots,N$.
Each $X^i_t$ takes values in the set $\mathds{R}_+$ of nonnegative real numbers.
Reserves level $X^i_t$ decreases at a controlled production rate $q^i_t \geq 0$,
and also has random discrete increment due to exploration.
We use a controlled point process $(N^i_t)$ to model arrivals of new discoveries.
Specifically, $N^i_t$ has intensity $\lambda(t) a^i_t$, where $a^i_t$ is the exploration effort controlled by player $i$.
The parameter $\lambda(t)$ is rate of discovery per unit exploration effort which reflects the current exploration techniques and overall resources underground, it is thus taken as exogenously given and the same for all producers.
Since  total resources underground are depleted over time,
it is reasonable to assume that $\lambda(t)$ is decreasing in $t$ and
$\lim_{t\rightarrow \infty} \lambda(t) = 0$.
Let $\tau_n^i$ be the $n$-th arrival time of the point process $N_t^i$,
then the inter-arrival time between the $n$-th and $(n+1)$-st arrivals satisfies the following probability distribution
\[
\mathds{P} \left( \tau^i_{n+1} > \tau^i_{n}+t \right) =
\exp \left( - \int_{\tau^i_n}^{\tau^i_{n}+t} \lambda(s) a^i_s d s \right) .
\]
Let $\delta$ denote the unit amount of a discovery,
which is assumed to be a positive constant as in \cite{LS-Cournot,LudkovskiYang14}.
The unit amount of discovery $\delta$ can be random in general,
see Remark \ref{remark: random delta}.
The reserves dynamics of each player are given by the following stochastic differential equation
\begin{align}
\label{individual reserves process}
dX^i_t & = -q^i_t \mathds{1}_{\{X^i_t>0\}} dt + \delta d N^i_t , \quad i = 1,2,\ldots,N , \qquad	
X^i_0 = x^i_0 \geq 0,
\end{align}
where $x^i_0$ is the initial reserves level. The indicator function $\mathds{1}_{\{X^i_t>0\}}$ means that production must be shut down, $q^i_t = 0$, whenever reserves run out, $X^i_t=0$. With \eqref{individual reserves process} reserves decrease continuously between discoveries
according to the production schedule and experience  an instantaneous jump of size $\delta$ at discovery epochs.

\textbf{Cost functions.} We assume that all producers have identical quadratic cost functions of production and exploration,
denoted by $C_q(\cdot)$ and $C_a(\cdot)$ respectively,
\begin{equation}
\label{production and exploration cost functions}
C_q(q) = \kappa_1 q + \beta_1 \frac{q^2}{2}, \qquad
C_a(a) = \kappa_2 a + \beta_2 \frac{a^2}{2} .
\end{equation}
The coefficients $\beta_{1,2} > 0$ of the quadratic terms are assumed to be positive, making the cost functions strictly convex and guaranteeing that the optimal production and exploration effort levels are finite.
The coefficients $\kappa_{1,2} \ge 0 $ of the linear terms represent constant marginal cost of production and exploration due to the use of facilities and labor, while $\beta_{1}, \beta_2$ of the quadratic terms represent increasing marginal costs due to negative externalities (such as rising labor costs or nonlinear taxation). We note that when $\kappa_2 =0$ then exploration is always undertaken, otherwise $a^\ast_t = 0$ could be optimal.

\textbf{Supply-demand equilibrium. }
We assume there is a single market price $p$ (so the market is undifferentiated which is a reasonable assumption for the energy industry); in equilibrium $p$ is determined by equating the total demand to the total supply at that price level. This equilibrium is achieved instantaneously at each date $t$, which is viewed as fixed in the following exposition.

In addition to the $N$ producers we assume there  are $M$  undifferentiated consumers.
The demand of consumer $j$, denoted by $d^j$, depends on the price through the linear demand function $d^j = D(p)= L - p$. Note that demand is finite even at zero price and the demand elasticity is bounded.
The aggregate demand, denoted by $Q^{(M)}_{demand}$, is the sum
$$
Q^{(M)}_{demand} := \sum_{j=1}^M d^j = M(L-p).
$$
We now equate total demand with total supply $Q^{(N)}$ and substitute the right-hand-side above to obtain the equilibrium relation between total supply and price,
$
Q^{(N)} = M(L-p).
$
Finally, the clearing price can be represented through the inverse demand function
\begin{equation}
p_t = L - \frac{1}{M} Q^{(N)}_t 
= L - \frac{N}{M} \left( \frac{1}{N}\sum_{i=1}^N q^i_t \right)  = L - \frac{N}{M} \bar{Q}_t,
\label{inverse demand function}
\end{equation}
where $\bar{Q}_t$ is the \emph{average} production rate. To obtain a nontrivial limiting price as the number of producers and the number of consumers both go to infinity, we see that it is necessary to take $M \propto N$. Without loss of generality (if necessary by redefining $L$), we assume $M=N$, so that $p_t= L - \bar{Q}_t= D^{-1}(\bar{Q}_t)$ where $L$ can be interpreted as the cap on prices as supply vanishes.

\subsection{Game value functions and strategies}
In a continuous-time Cournot game model, each player continuously chooses rate of production $q^i_t$ in order to maximize profit which is equal to the revenue $p_t\cdot q^i_t$, minus the production and exploration costs,
integrated and discounted at a rate $r>0$. We work on a finite time horizon $[0, T]$, where $T$ is exogenously specified. The role of the horizon will be revisited in the sequel.
The price $p_t$ each player receives is determined through the inverse demand function \eqref{inverse demand function}.

Denoting player-$i$'s strategy by $s^i :=\left( q^{i}, a^{i} \right)$, the overall strategy profile for all the $N$ players is $\boldsymbol{s}:= \left( s^1, s^2, ... , s^N \right)$.
Starting with reserves state $\left( X^1_t, \cdots , X^N_t \right) =: \bf{X}_t = \bf{x}$, each player's objective functional $\mathcal{J}^i, i=1,\ldots,N$ on the horizon $[t,T]$ is defined as the total discounted profit
\begin{align}
&
\mathcal{J}^i \left( \boldsymbol{s} ; t, \bf{x} \right) :=
 \E
\left\{ \int_t^T \left[ D^{-1}\left(  \frac{1}{N} \sum_{j=1}^N q^j_s \right) q_s^i - C_q(q_s^i) - C_a(a_s^i) \right] e^{-r(s - t)} \ ds   \ \Big| \ \bf{X}_t = \bf{x}
\right\},
\label{N players objective functionals}
\end{align}
where the expectation is over the random point processes $N^j$'s that drive $X^j$'s and hence $q^j$'s.
We focus on the admissible set  $\mathcal{A}$ of strategies whereby $s^i_t = (q^i_t, a^i_t)$ are Markovian feedback controls  $q^i_t=q^i(t, {\bf X}_t)$, $a^i_t=a^i(t, {\bf X}_t)$ such that $\mathcal{J}^i({\bf{s}}; t, \bf{x}) < \infty$,
$\forall {\bf{x} } \in \mathds{R}^N_{+}$,
for all $i = 1, \ldots, N$. 
%
From \eqref{N players objective functionals}, we see that each player's choice of strategy depends on the strategies of all the others, leading to a non-cooperative game. Our aim is to investigate the resulting (Markov feedback) Nash equilibrium. Importantly, the feedback structure of the controls
together with \eqref{inverse demand function} imply that player $i$'s dependence on $\bf{X}_t$ can be summarized by his individual reserves
$X^i_t$ and the aggregate distribution of all players' reserves. The latter is characterized through the upper-cumulative distribution function defined by
$$
\eta^{(N)}(t, x) := \frac{1}{N}\sum_{j=1}^N \mathds{1}_{ \left\{ X^j_t \geq x \right\} }.
$$
Thus the Markovian feedback controls $\left( q^i, a^i \right)$ can be equivalently represented as
\begin{align}
q^i_t &
= q^i\left( t, X^i_t ; \eta^{(N)}(t,\cdot) \right) , \qquad
a^i_t
= a^i\left( t, X^i_t;  \eta^{(N)}(t,\cdot) \right) ,  \qquad
i=1,\ldots,N. \notag
\end{align}

\begin{defn}[Nash equilibrium]
A Nash equilibrium of the $N$-player game is a strategy profile
$\boldsymbol{s}^\ast = \left( s^{1,\ast},\ldots,s^{N,\ast} \right)$,
with $s^{i,\ast} := \left( q^{i, \ast}, a^{i, \ast} \right) $
such that
\begin{align}
\mathcal{J}^i\left( \boldsymbol{s}^\ast;  t, \bf{x} \right)
\geq
\mathcal{J}^i\left(  (\boldsymbol{s}^{\ast, -i}, s^i) ; t,  \bf{x} \right) , \quad
\forall i \in \{ 1, 2, \ldots, N \},
\end{align}
where $(\boldsymbol{s}^{\ast,-i}, s^i)$ is the strategy profile
$\boldsymbol{s}^\ast$ with the $i$-th entry
 replaced by arbitrary $s^i=(q^{i}, a^{i}) \in \mathcal{A}$.

\end{defn} 
In words, a Nash equilibrium is the set of strategies of the $N$ players such that no one can better off by unilaterally changing his own strategy.  Theoretically the Nash equilibrium of the $N$-player game can be found by Hamilton-Jacobi-Isaacs (HJB-I) approach.
HJB-I approach is to use dynamic programming principle to derive the partial differential equation of each player's game value function, with other players' strategies as entries.
It is extremely hard to find a Nash equilibrium by using the (HJB-I) approach either analytically or numerically, even for small $N$, e.g. $N=2$.
Thus in the following Section~\ref{Mean field game problem as N infinity}, we introduce the
mean field game model as $N \rightarrow \infty$, which serves as an approximation to the Nash equilibrium of the game when number of players is very large.

\subsection{Mean field game problem as $N \rightarrow \infty$}
\label{Mean field game problem as N infinity}

As the number of players becomes very large $N \rightarrow \infty$, thanks to the Law of Large Numbers,
the empirical distribution $\eta^{(N)}$ is expected to converge to a  CDF $\eta$.
The limiting $\eta(t, x)$ is regarded as the reserves distribution among all players at date $t$, which means, for a given $t$ and $x$, the proportion of all players at time $t$ with reserves level greater than or equal to $x$. The production and exploration controls continue to take the Markovian feedback form
\begin{align}
\label{production and exploration controls mean field game}
q_t &=q(t, X_t; \eta(t,\cdot)), 	\quad
a_t =a(t, X_t; \eta(t,\cdot)).
\end{align}
To re-solve for the supply-demand equilibrium clearing price, we use the total quantity $Q(t)$ of production at time $t$, defined as the Stieltjes integral of a representative producer's production rate with respect to the reserves distribution,
\begin{equation}
Q(t) := - \int_0^\infty q(t, x; \eta) \ \eta(t, dx) .
\label{mfg total production definition}
\end{equation}
Note that $\eta(t,x)$ is decreasing in $x$, thus we add a negative sign to the integral in order to keep $Q(t)$ positive. The definition in \eqref{mfg total production definition} is the limit of the original $\bar{Q}_t$ that was defined for the $N$-player game. As before $Q(t)$ is linked to the clearing price  via
\begin{equation}
p(t, \eta(t,\cdot)) = D^{-1}(Q(t))= L + \int_0^\infty q(t, x; \eta) \ \eta(t, dx) .
\label{mfg price function}
\end{equation}

For a representative producer who starts with initial reserves level $X_t=x$ (with $x \in \mathbb{R}_+$ now a scalar), and representing all other players' states via $\eta(\cdot, \cdot)$  taken as given, the mean-field objective functional is defined analogously to \eqref{N players objective functionals}:
\begin{align}
\label{mfg objective functional}
 \mathcal{J} \left( q, a ; t, x, \eta \right) 	
 :=
 \E
\left\{ \int_t^T \left[ p(s,\eta(s,\cdot))
q_s - C_q(q_s) - C_a(a_s) \right] e^{-r(s-t)} \ ds \
\middle| X_t = x\right\} .
\end{align}
Above the strategies $(q_t, a_t)$ take the Markovian feedback form
\eqref{production and exploration controls mean field game} and the reserves distribution $(\eta(s,\cdot))$ is a probability upper-CDF for all $s \in [t,T]$.  We again remark that
the profit of a player depends on all the other players through the mean field term $Q(s)$.
We define the mean field game Nash equilibrium of our model as
\begin{defn}[Mean field game Markov Nash equilibrium]
\label{Mean field game Markov Nash equilibrium definition}
A MFG MNE is a triple $\left( q^\ast, a^\ast, \eta^\ast \right)$ of adapted processes on $[0,T]$ such that,
denoting by $X^\ast_t$ the solution of
\begin{equation}
\label{reserves dynamics mean field game case}
dX^\ast_t = -q^\ast(t,X^\ast_t, \eta^\ast(t,\cdot) ) \, dt + \delta d N^\ast_t , \quad t \geq 0 , \qquad X^\ast_0 \sim \eta^\ast(0,\cdot),
\end{equation}
then $\eta^\ast(t, x) = \PP( X^\ast_t \ge x)$ is the upper-CDF of $X^\ast_t \, \forall t \in [0, T]$ and
\begin{equation}
\label{mfg optimality condition}
\mathcal{J}(q^\ast, a^\ast; t, x, \eta^\ast) \geq \mathcal{J}(q, a; t, x, \eta^\ast),
\
\forall (q, a)\in \mathcal{A} .
\end{equation}

\end{defn}

Definition \ref{Mean field game Markov Nash equilibrium definition} consists of two conditions. One condition, which we can call optimality condition, is that each producer chooses strategy $(q^\ast, a^\ast)$ which gives optimal game value, given the others' strategies.
The second condition, which we can call consistency condition, is that the reserves dynamics of each player under the control of the strategy $(q^\ast, a^\ast)$ has the upper cumulative distribution function $\eta^\ast$ that is the same as the one that enters the objective functional.
%
In Section \ref{sec: Mean field game Nash equilibrium}, we introduce the differential equations characterizing  the MFG MNE defined in Definition
\ref{Mean field game Markov Nash equilibrium definition}, which is the core problem of this paper.

\section{Mean field game Nash equilibrium}
\label{sec: Mean field game Nash equilibrium}

Solving for MFG MNE involves two partial differential equations.
One equation is the HJB equation of the game value function of a representative producer, which is derived by a dynamic programming principle and yields the equilibrium production and exploration strategies $(q^\ast, a^\ast)$.
The other equation is the transport equation characterizing the distribution $\eta^\ast$ of reserves process $X^\ast$ controlled by the strategies $(q^\ast, a^\ast)$ obtained from the HJB equation.

Section \ref{sec: Game value function of a representative player}, treats the HJB equation associated to the game value function of a representative producer.
The PDE that characterizes the evolution of the reserves distribution will be discussed in Section~\ref{sec: Transport equation of reserves distribution}.
The overall coupled
system associated to the MFG MNE is taken up in Section~\ref{sec: System of doubly coupled HJB-transport equations}.
Details of numerical methods and examples will be discussed in Section
\ref{sec: Numerical methods and examples}.

\subsection{Game value function of a representative producer}
\label{sec: Game value function of a representative player}

Let us fix a sequence of probability CDF's $\eta(t,\cdot)$.
Associated with the objective functional \eqref{mfg objective functional}, we define the game value function $v^\eta(t,x )$ of a representative producer by
\begin{align}
\label{mfg value function}
v^\eta(t,x)  &	:= \sup_{(q, a) \in \mathcal{A} } \mathcal{J} \left( q, a ; t, x, \eta \right)	\notag \\
&
= \sup_{ (q, a) \in \mathcal{A} } \E
\left\{ \int_t^{T}\left[ p(s; \eta) q_s - C_q(q_s) - C_a(a_s) \right] e^{-r(s-t)} ds \ \middle| \ X_t = x \right\} ,
\end{align}
where the player chooses her production rate $q(t, X_t; \eta)$ and exploration rate
$a(t, X_t; \eta)$ from the set $\mathcal{A}$ of Markovian feedback controls \eqref{production and exploration controls mean field game}. Note that above $\eta$ is treated as an exogenous parameter, while the price  is still endogenous being a function of total production: $p(t; \eta(t,\cdot)) = D^{-1}\left(Q(t) \right)$ as in \eqref{mfg price function}. This introduces a global dependence between the map $x \mapsto q(t,x)$ and $p(t)$.

Define the forward difference operator $\Delta_x$ as
$\Delta_x v(t, x) :=v(t, x+\delta) - v(t, x)$.

\begin{lemma}
\label{lemma: mean field game explicit HJB equation}
The game value function $v^\eta(t, x)$ defined by \eqref{mfg value function} satisfies the HJB equation
\begin{align}
\label{mfg HJB}
0 & = \frac{\partial}{\partial t} v^\eta(t,x)  -r v^\eta(t,x)
+ \frac{1}{2\beta_1} \left[ \left(p(t; \eta(t,\cdot)) - \kappa_1  - \frac{\partial}{\partial x} v^\eta(t,x) \right)^+ \right]^2 	\notag 	\\
& \quad + \frac{1}{2\beta_2} \left[ ( \lambda(t) \Delta_x v^\eta(t,x) - \kappa_2 )^+  \right]^2 ,
\end{align}
with terminal condition $v^\eta(T,x) = 0$,
where the optimal $q^\eta(t, x)$ and $a^\eta(t, x)$ are given by
\begin{align}
& q^\eta(t,x)
= \frac{1}{\beta_1} \left( L  - Q^\eta(t) - \kappa_1
- \frac{\partial}{\partial x} v^\eta(t, x) \right)^+
 , 	\label{eq:q-star}	 \\
& a^\eta(t, x)
 =  \frac{1}{\beta_2}\left( \lambda(t) \Delta_x v^\eta(t,x) - \kappa_2 \right)^+ ,
\label{optimal mfg q and a explicit}
\end{align}
with $Q^\eta(t)$ uniquely determined by the equation
\begin{align}
Q^\eta(t)
&
+ \int_0^\infty
\frac{1}{\beta_1}\left( L - \kappa_1 - \frac{\partial}{\partial x} v^\eta(t,x) - Q^\eta(t) \right)^+
\eta(t, dx) = 0. \label{eq:Q-star}
\end{align}
The price $p(t)$ depends on $q^\eta(t,\cdot)$  and the given reserves distribution $\eta(t,\cdot)$ via \eqref{mfg price function}.

\end{lemma}

\begin{proof}
The associated HJB equation of \eqref{mfg value function} derived by the dynamic programming principle, is
\begin{multline}
\label{eq:direct-hjb}
  0 = \frac{\partial}{\partial t} v^\eta(t,x)  - r v^\eta(t,x)
+ \sup_{a \geq 0} \left[ - C_a(a) +a\lambda(t) \Delta_x v^\eta(t,x) \right] 	 \\
+ \sup_{q\geq 0}\left[ p(t; \eta(t,\cdot))  q - C_q(q)
- q \frac{\partial}{\partial x} v^\eta(t,x) \right]   ,
\end{multline}
where the forward difference term $\Delta_x v(t, x)$ is due to the jumps in the reserves dynamics, cf.~\cite{LS-Cournot}.
The optimal exploration rate $a^\eta$ is determined by the first order condition
\begin{align}
 a^\eta(t, x) &  =
\argmax_{ a \geq 0}
\left[ - C_a(a) +a\lambda(t) \Delta_x v^\eta(t,x) \right]
 =  \frac{1}{\beta_2}\left( \lambda(t) \Delta_x v^\eta(t,x) - \kappa_2 \right)^+ ,
\label{eq:a-foc}
\end{align}
where we plugged the quadratic form of $C_a$ from \eqref{production and exploration cost functions}.
Similarly, maximizing the last term in \eqref{eq:direct-hjb} to solve for
the optimal production rate $q^\eta$ leads to the first order condition
\begin{align}
0 & =
\frac{\partial}{\partial q}
 \left[ p(t,\eta(t,\cdot))  q^\eta(t,x)
- C_q(q^\eta(t,x)) - q^\eta(t,x) \frac{\partial}{\partial x} v^\eta(t,x) \right]	\notag	\\
\Leftrightarrow \quad \beta_1 q^\eta(t,x) & = \left( p(t,\eta(t,\cdot))  - \kappa_1 - \frac{\partial}{\partial x} v^\eta(t,x) \right)^+.
 	\label{eq:q-foc}
\end{align}
Using $p(t,\eta(t,\cdot)) =  L - Q^\eta(t)$ yields \eqref{eq:q-star}. 
Integrating the right-hand side of \eqref{eq:q-star} with respect to $\eta(t, \cdot)$,
\begin{align}
 \int_0^\infty
\frac{1}{\beta_1}\left( L - \kappa_1 - \frac{\partial}{\partial x} v^\eta(t,x) - Q^\eta(t) \right)^+
\eta(t, dx)  =  \int_0^\infty q^\eta(t,x) \eta(t, dx) = -Q^\eta(t).
\label{eq:integral q eta}
\end{align}
Thus, $Q^\eta(t))$ satisfies $G(Q^\eta(t)) = 0$ as in \eqref{eq:Q-star} where
\begin{align*}
G( Q ) = Q+ \int_0^\infty
\frac{1}{\beta_1}\left( L - \kappa_1 - \frac{\partial}{\partial x} v^\eta(t,x) - Q \right)^+
\eta(t, dx) .
\end{align*}
Assuming $L > \kappa_1$ (otherwise production is never profitable and $Q(t) = 0$), we have $G(0) > 0$
and $G(  L-\kappa_1  )  < 0$. Moreover, $Q \mapsto G(Q)$ is continuous
because the integrand
is uniformly bounded,
$\left| \left( L - \kappa_1 - \frac{\partial}{\partial x} v^\eta(t,x) - Q  \right)^+ \right| \leq (L - \kappa_1)$. Since $Q \mapsto G(Q)$ is decreasing it follows
that a unique root $Q(t)$ exists  in $[0, L-\kappa_1]$. Finally \eqref{mfg HJB} follows by using \eqref{eq:a-foc} and \eqref{eq:q-foc} in \eqref{eq:direct-hjb}.
\end{proof}

We observe two non-standard features of the HJB equation \eqref{mfg HJB}. First, the optimal production control \eqref{eq:q-star} does not only depend on the individual producer's value function $\frac{\partial}{\partial x} v^\eta(t,x)$, but also on the reserves distribution of all the players through the mean field term
$\int_0^\infty \frac{\partial}{\partial z} v^\eta(t,z) \eta(t, dz)$.
Second, \eqref{mfg HJB} contains two \emph{non-local} terms: the forward difference $v^\eta(t, x+\delta) - v^\eta(t,x)$ and the integral
$\int_0^\infty \frac{\partial}{\partial z}v^\eta(t, z) \eta(t, dz)$.

The HJB equation has two boundary conditions. At $t=T$ we take $v(T, x) = 0$ as no more production is assumed possible beyond the prescribed horizon. Furthermore, the exhaustibility constraint $x \ge 0$ imposes a  boundary condition at $x=0$ similar to the model in \cite{LS-Cournot} for a single exhaustible producer.
Since production $q(t,0)=0$ is zero on the boundary $x=0$, we have
\begin{align}
0&= \frac{\partial}{\partial t} v^\eta(t,0)  - r v^\eta(t,0)
 +  \sup_{a\geq 0} \left[ - C_a( a ) + a \lambda(t) \Delta_x v^\eta(t,0) \right]
\notag	\\
&= \frac{\partial}{\partial t} v^\eta(t,0)  - rv^\eta(t,0)
+ \frac{1}{2\beta_2} \left[ ( \lambda(t) \Delta_x v^\eta(t,0) - \kappa_2 )^+  \right]^2 ,
\qquad 0\leq t < T.
\label{HJB equation on boundary}
\end{align}
We will use the boundary condition equation \eqref{HJB equation on boundary} in the numerical schemes. At the other extreme, as $x \to \infty$ then $\Delta_x v^\eta(t,x) \to 0$ and hence for $\kappa_2 > 0$ we have $a^\eta(t,x) = 0$ from \eqref{eq:a-foc}. Thus, there is a saturation reserves level $x_{sat}(t)$ \cite{LS-Cournot,LudkovskiYang14} such that $a^\eta(t,x) = 0 \forall x \ge x_{sat}(t)$: with a lot of reserves and a strictly positive marginal cost, exploration becomes unprofitable (furthermore, since $v^\eta(t,x)$ is expected to be concave in $x$, $a^\eta(t,x)$ is monotone decreasing).

\subsection{Transport equation of reserves distribution}
\label{sec: Transport equation of reserves distribution}

In this section~we study evolution of the reserves distribution through the transport equation of the upper-cumulative distribution function
$\eta(t, \cdot)$ of the reserves process $X_t$ from \eqref{individual reserves process}
where $N_t$ is a point process with controlled rate $\lambda(t) a_t$,
and the production rate $q_t = q(t, X_t)$ and exploration rate $a_t = a(t, X_t)$ are given, i.e.~treated as exogenous inputs.

When reserves reach zero $X_t=0$, production shuts down $q_t = 0$.  With exploration effort being made, the reserves level $X_t$ can bounce back to $X_{\tau}=\delta$, however the waiting time until next discovery is strictly positive. As a result, $\PP(X_t = 0) >0$, i.e.~the distribution of $X_t$ has a point mass at $x=0$.
Thus to study the evolution of the distribution of $X_t$,
we consider two parts:
 the upper-cumulative distribution function $\eta(t, x) = \PP \left( X_t \geq x \right)$
in the interior $x>0$; and the boundary probability $\pi(t):= \PP(X_t = 0)=1-\eta(t,0+)$.
The upper-CDF $\eta(t, x)$ is regarded as the proportion of players with reserves level greater than or equal to $x$,
and $\pi(t)$ is interpreted as the proportion of producers with no reserves. The following proposition gives the piecewise PDE that 
$\eta$ satisfies. See the proof in Appendix \ref{Proof of proposition transport equation}. Observe that because production slows down as reserves are exhausted $\lim_{x \downarrow 0} q(t,x) = 0$, there is no boundary condition for $\eta$ at $x=0$; instead $\pi(t)$ shows up in the PDE for $\eta$.

\begin{prop}[Transport equation]
\label{prop: transport equation}

The distribution of the reserves process $X_t$ is characterized by the pair $\left( \pi(t), \eta(t, x) \right)$, where $\eta(t,x)= \PP(X_t \geq x)$, $0<x<\infty$, and $\pi(t) := 1-\eta(t,0+)$:
\begin{subequations}
\begin{align} 
 \label{eq:transport-2}
\frac{\partial}{\partial t} \eta(t,x) &= \lambda(t) a(t,0) \pi(t) - \int_{0+}^x \lambda(t) a(t,z)\eta(t, dz)
+ q(t,x)\frac{\partial}{\partial x} \eta(t,x) , \quad 0<x \leq \delta ;  \\
\frac{\partial}{\partial t} \eta(t,x) &= - \int_{x -\delta}^x \lambda(t) a(t,z)
\eta(t, dz)
+ q(t,x) \frac{\partial}{\partial x} \eta(t,x), \quad x> \delta . \label{eq:transport-3}
\end{align}\label{transport equation} 
\end{subequations}
with given initial condition $\eta(0, x) = \eta_0(x)$ and $\pi(0)= 1-\eta_0(0+)$.
\end{prop}

The discontinuity of $\eta(t,\cdot)$ at $x=0$ generates higher order discontinuities at $x=\delta, 2\delta, 3\delta, \cdots$. Indeed, at $x = k\delta$ only the first $(k-1)$ derivatives of $\eta(t,x)$ exist. In other words, the distribution of $X_t$ has a point mass at $x=0$, a first-order discontinuity (non-continuous density) at $x=\delta$ and a smooth density for all other $x > 0$. This non-smoothness is the reason why we do not work with the ill-defined density ``$m(t,x) = -\frac{\partial}{\partial x} \eta(t,x)$''.

\begin{remark}
\label{remark: random delta}
The size of  new discoveries $\delta$ can be random in general.
We may model discovery amounts via a stochastic sequence
$\delta_n, n=1, 2, \ldots,$
where each $\delta_n$ is identically distributed with some distribution $F_\delta(\cdot)$
and independent of everything else in the model.
Introducing $F_\delta$ entails replacing the integral $\int_{x-\delta}^x \lambda(t) a(t,z) \eta(t, dz)$  in \eqref{eq:transport-3} with
$\int_0^x F(du) \int_{x-u}^x \lambda(t) a(t, z) \eta(t, dz)$. Similarly, in the HJB equation we would replace $v(t, x+\delta)$ with $\int_0^\infty v(t, x+u) F_\delta(du)$.
For simplicity we stick to fixed discovery sizes for the rest of the article.
\end{remark}

\subsection{System of HJB-transport equations}
\label{sec: System of doubly coupled HJB-transport equations}

The consistency condition of Definition \ref{Mean field game Markov Nash equilibrium definition} implies that a MFG MNE is characterized by the HJB equation \eqref{mfg value function} where we plug-in the equilibrium CDF $\eta^\ast$, and the transport equation \eqref{transport equation} where we plug-in the equilibrium $q^\ast$ and $a^\ast$. The equilibrium  price  process is $p^\ast(t) = L + \int_0^\infty q^\ast(t, z) \eta^\ast(t, dz)$. The resulting system is summarized in the following.

\begin{prop}[MFG PDE's]
\label{prop: Mean field game partial differential equations}
The mean field game Nash equilibrium $(q^\ast, a^\ast, \eta^\ast)$ is determined by the HJB equation:
\begin{align}
\label{mfg HJB equation}
0&= \frac{\partial}{\partial t} v(t,x)  - r v(t,x)
 +  \left[ - C_a( a^\ast(t,x) ) + a^\ast(t,x)\lambda(t) \Delta_x v(t,x) \right]  \notag \\
&
\quad + \left[ p^\ast(t) q^\ast(t,x) - C_q(q^\ast(t,x))
- q^\ast(t,x) \frac{\partial}{\partial x} v(t,x) \right],   \quad
0< x , \ 0\leq t < T ,
\end{align}
where the $q^\ast(t, x)$ and $a^\ast(t, x)$ are given by
\begin{align}\label{eq:q-star-2}
& q^\ast(t,x)
= \frac{1}{\beta_1} \left( L  - Q(t) - \kappa_1
- \frac{\partial}{\partial x} v(t, x) \right)^+ ,
\\ \label{eq:a-star-2}
& a^\ast(t, x)
 =  \frac{1}{\beta_2}\left( \lambda(t) \Delta_x v(t,x) - \kappa_2 \right)^+ ,
\end{align}
with $Q(t)$ uniquely determined by the equation
\begin{align}
Q(t)
& =
- \int_0^\infty
\frac{1}{\beta_1}\left( L - \kappa_1 - \frac{\partial}{\partial x} v(t,x) - Q(t) \right)^+
\eta^\ast(t, dx) = - \int_0^\infty q^\ast(t, x) \eta^\ast(t,d x), 	\label{total production}
\end{align}
and the transport equation:
\begin{subequations} 
\begin{align}
\frac{\partial}{\partial t} \eta^\ast(t,x) &= \lambda(t) a^\ast(t,0) (1-\eta^\ast(t,0+)) - \int_{0+}^x \lambda(t) a^\ast(t,z)\eta^\ast(t, dz)
+ q^\ast(t,x)\frac{\partial}{\partial x} \eta^\ast(t,x) , \quad 0<x \leq \delta ;  \label{eq:eta-1}  \\
\frac{\partial}{\partial t} \eta^\ast(t,x) &= - \int_{x -\delta}^x \lambda(t) a^\ast(t,z)
\eta^\ast(t, dz)
+ q^\ast(t,x) \frac{\partial}{\partial x} \eta^\ast(t,x), \quad\qquad x > \delta . \label{eq:eta-2}
\end{align} 
\label{mfg transport equation}
\end{subequations}
\end{prop}

The HJB equation and transport equation are doubly coupled with  $\eta^\ast$ entering the HJB equation through the aggregate production which is an integral of optimal production rates $q^\ast(t,x)$ with respect to the mean-field reserves distribution $\eta^\ast(t,dx)$. Conversely, the optimal production and exploration rates $(q^\ast, a^\ast)$ obtained from the HJB equation of a representative producer drive the reserves distribution $\eta^\ast(\cdot)$.

Existence, uniqueness, and regularity of the solutions of the system of MFG PDE's is still an ongoing challenge and an area of active research. For the system \eqref{mfg HJB equation}--\eqref{mfg transport equation} the difficulty in proving existence and uniqueness of solutions lies in the non-local coupling term
$\int_0^\infty q(t,x) \eta(t, dx)$ and the forward delay term $\Delta_x v(t, x) = v(t, x+\delta) - v(t,x)$. In the more common \emph{local} coupling situation, the mean-field interaction for a representative producer with state $(t,x)$ is of the form $F(t, x, m(t, x))$, i.e.~the player interacts with the density of her neighbors $m(t,x)$ at the same $(t,x)$. In contrast, in the supply-demand context, the interaction includes \emph{all} players, namely their production rates (that can be linked to the marginal values $\frac{\partial}{\partial x} v(t, x) $) across all $x$.

Related proofs for second-order Cournot MFG PDEs have been provided in \cite{GraberBensoussan15, GraberMouzouni17}. The respective reserves dynamics involve Brownian noise and no jump terms (no exploration). Specifically, %
Graber and Bensoussan \cite{GraberBensoussan15} established existence and uniqueness of MFG MNE in the case that players leave the game after exhaustion (Dirichlet boundary conditions), while
\cite{GraberMouzouni17} recently proved existence and uniqueness of solutions in the case where reserves can be exogenously infinitesimally replenished at $x=0$ (corresponding eventually to Neumann boundary conditions). Their model (with zero volatility) can be viewed as the non-exploration $\lambda \equiv 0$ sub-case of our model. However, exogenous discoveries imply that the reserves distribution is a probability density on $(0,X_{max})$, obviating the need to track $\pi(t)$ which significantly simplifies the respective proof.
 In a related vein, Cardaliaguet and Graber~\cite{CardaliaguetGraber15} gave detailed proof of existence and uniqueness of equilibrium solution for first order MFG's with local coupling. However, first order MFG PDEs with non-local terms $v(t, x+\delta) - v(t, x)$  to the best of our knowledge have not been discussed in existing literature (except in passing in \cite[Sec 5]{ChanSircar16}), and the respective existence, uniqueness, and regularity of solutions remain an open problem. 

The MFG framework links the individual strategic behavior of each producer with the macro-scale organization of the market. Therefore the main economic insights concern the resulting \emph{aggregate} quantities that describe the overall evolution of the market. For this purpose, we recall the total production $Q(t)$  defined in \eqref{total production}
$A(t)$ the total discovery,
and $R(t)$ the total reserves, which are defined respectively as
\begin{align}
R(t) & = \int_0^\infty \eta^\ast(t, x) dx, 		\label{total reserves}		\\
A(t)	 & = -\delta \int_0^\infty \lambda(t) a^\ast(t,x) \eta^\ast(t, d x).	\label{total discovery}
\end{align}
Note that $R(t) = \int_0^\infty \PP( X_t \ge x) dx = \E[ X_t]$ justifying its meaning of total reserves.
The following Lemma \ref{lemma: relation of Q A R}, proven in Appendix \ref{app: relation of Q A R}, shows the relation between these quantities of interest. It can be interpreted as conservation of mass for the reserves: at the macro-scale total reserves change is simply the net difference between reserves additions (via new discoveries $A(\cdot)$) and reserves consumption (via production $Q(\cdot)$).
\begin{lemma}
\label{lemma: relation of Q A R}
We have the relation
\begin{align}
\frac{d}{dt}R(t) & = -Q(t) + A(t)	,
\label{relation of Q A R differential form}	\quad\text{i.e.~}\quad
R(t)
 = R(0) - \int_0^t Q(s) \, ds + \int_0^t A(s) \, ds .	
\end{align}

\end{lemma}


\section{Numerical methods and examples}
\label{sec: Numerical methods and examples}

We use an iterative scheme to numerically solve the system of HJB equation
\eqref{mfg HJB equation} and transport equation \eqref{mfg transport equation},
similar to the approach in  \cite{GLL10,ChanSircar14}.
The Picard-like iterations start with an initial price process $p^{(0)}(\cdot)$ as an input into the MFG value function \eqref{mfg HJB}, which reduces to a standard optimization problem for the production and exploration rates $(q^{(0)}, a^{(0)})$.
Then we input  $(q^{(0)}, a^{(0)})$ into the equation
\eqref{transport equation} of reserves evolution to solve for  $\eta^{(0)}(\cdot, \cdot)$.
The $q^{(0)}$ and $\eta^{(0)}$ obtained are used to update the price \eqref{mfg price function}, via
$p^{(1)}(t) = \frac{1}{2}\left[ D^{-1}\left( -\int_0^\infty q^{(0)}(t, x) \eta^{(0)}(t, dx)\right) + p^{(0)}(t) \right] $.
The updated price $p^{(1)}(\cdot)$ is then used for a new iteration. As $k \to \infty$, the iterations are expected to converge to a fixed point, i.e.~a
triple $(q^{(\infty)}, a^{(\infty)}, \eta^{(\infty)})$ that simultaneously satisfies the HJB equation \eqref{mfg HJB equation} and transport equation \eqref{mfg transport equation} and hence yields a MFG MNE.

For numerical purposes we restrict to a bounded space domain $[0,X_{max}]$ which is further partitioned using a mesh $0=x_0<x_1< ... < x_M= X_{max}$, with equal mesh size
$\Delta x = x_m - x_{m-1}, m=1,\ldots,M$. Below we fix $\Delta x = 0.1$ in all the computational examples. Other numerical parameters of our examples are summarized in Table
\ref{Parameters values for numerical analysis}.


\begin{table}[htb]
\centering
$$\begin{array}{l|l} 
\text{ Cost Functions} & \kappa_1 = \kappa_2 = 0.1, \quad \beta_1 = \beta_2 = 1 \\
\text{ Max Price/Int Rate} & L = 5, \quad r= 0.1 \\
\text{ Reserves dynamics } & \delta = 1, \quad \lambda = 1 \\
\text{ Numerical Scheme} & T=50, X_{max} = 120, \Delta x = 0.1 \\ \hline
\end{array} $$
\caption{Parameter values used for all numerical illustrations in Section \ref{sec: Numerical methods and examples}.}
\label{Parameters values for numerical analysis}
\end{table}

In Section~\ref{sec: Numerical method for HJB equation}, we introduce the numerical method to solve the HJB equation of a representative producer's game value function with price $p(t)$ exogenously given.
In Section~\ref{sec: Numerical method for transport equation}, we introduce the numerical method to solve the equation \eqref{transport equation} of reserves distribution controlled by the optimal $(q, a)$ obtained in the previous step.
In Section~\ref{sec: Numerical method for the system of HJB and transport equations},
we show the iterative scheme to solve the coupled HJB and transport equations.

\subsection{Numerical scheme for the HJB equation}
\label{sec: Numerical method for HJB equation}

In this section~we solve for mean field game value function $v(t,x)$ defined by \eqref{mfg HJB} with an exogenously given price $p(t)$.   Treating $p(t)$ as exogenous allows us to avoid the production control formula in
\eqref{eq:q-star}	  
which has a mean-field dependence via $\int_0^\infty \frac{\partial}{\partial z} v(t, z)\eta(t, dz)$. Instead we use \eqref{eq:q-foc} that only depends on the player's own reserves state $x$, and reduces to a standard optimal stochastic control problem. For the exploration control we work with the first order condition as in \eqref{eq:a-foc}.
The HJB equation \eqref{mfg HJB} with boundary condition \eqref{HJB equation on boundary}
is similar to the single-agent problem in \cite{LS-Cournot}. The latter paper considered a time-stationary model which reduced the HJB equation to a first order nonlinear ordinary differential equation in $x$. In contrast, \eqref{mfg HJB} has time-dependence and hence is a genuine PDE.

We employ a method of lines to discretize the $x$ variable and treat the HJB PDE
as a system of ordinary differential equations in time variable $t$ with the terminal condition $v(T,x) = 0$.
The space derivative of $v(t, x)$ at each space grid point $x_m$ is approximated by a backward difference quotient
$\frac{\partial}{\partial x} v(t, x_m) \approx \frac{v(t, x_m) - v(t, x_{m-1})}{\Delta x}$.
The non-local term $\Delta_x v(t, x_m)$ is approximated by
$\Delta_x v(t, x_m) \approx v(t, x_{m+d}) - v(t, x_m)$ with $d = \lfloor \frac{\delta}{\Delta x} \rfloor$ so that $x_m + \delta \simeq x_{m + d}$.
We solve for $v(\cdot, x_m)$  as an ordinary differential equation in variable $t$, viewing $v(t, x_{m-1})$ and $v(t, x_{m+d})$ as source terms,
\begin{multline}
\frac{\partial}{\partial t} v(t, x_m)
\approx  r v(t, x_m)
 - \frac{1}{2\beta_1} \left[ \left( p(t) - \kappa_1  - \frac{v(t, x_m) - v(t, x_{m-1})}{\Delta x} \right)^+ \right]^2 \\
 - \frac{1}{2\beta_2} \left[ ( \lambda(t) [v(t, x_{m+d}) - v(t, x_m)]  - \kappa_2 )^+  \right]^2 ,
\quad m = 1, \ldots, M-d.
\label{eqn: ode vt interior}
\end{multline}
For the boundary case $m=0$, production stops and the equation becomes
\begin{align}
\frac{\partial}{\partial t} v(t, x_0)
= r v(t, x_0)
- \frac{1}{2\beta_2} \left[ ( \lambda(t) [v(t, x_{d}) - v(t, x_0)]  - \kappa_2 )^+  \right]^2 .
\label{eqn: ode vt left boundary}
\end{align}
Recall that for $x$ large enough, saturation level of reserves is reached and no exploration effort is made.
We take $X_{max}$ such that this would be true for $x_{M-d+1}, \ldots, x_M=X_{max}$ whereby the term
$( \lambda(t) \Delta_x v(t, x)  - \kappa_2 )^+$ vanishes and \eqref{eqn: ode vt interior} simplifies to
\begin{align}
\frac{\partial}{\partial t} v(t, x_m)
= r v(t, x_m)
- \frac{1}{2\beta_1} \left[ \left( p(t) - \kappa_1  - \frac{\partial}{\partial x} v(t, x_m) \right)^+ \right]^2 ,\quad m = M-d+1, \ldots, M.
\label{eqn: ode vt right boundary}
\end{align}
We use Matlab's Runge-Kutta solver \texttt{ode45} to solve (backward in time) the system
\eqref{eqn: ode vt interior}--\eqref{eqn: ode vt right boundary}
of ordinary differential equations for $\{v(t, x_m): m = 0, 1, \ldots, M\}$.

\subsubsection{A numerical example of the HJB equation }\label{sec:hjb-ex}

To illustrate the above approach to solve the HJB equation \eqref{mfg HJB}, we consider an example with a constant exogenous price $p(t)=3, \forall t \le T$. To prescribe $\lambda(t)$, observe that intuitively chances of a new discovery should be proportional to the remaining reserves underground. Assuming the global exploitable reserves decrease (linearly) in time due to ongoing exploration and production, we are led to consider a linear link between $t$ and discovery rate $\lambda(t)$:
\begin{align}
\lambda(t) = \left( 1 - t/\bar{T}\right)^+.
\notag
\end{align}
The time $\bar{T}$ can be viewed as global exhaustion of the commodity.

Figure \ref{fig: q_single_agent and a_single_agent} shows the resulting optimal production rate $q(t,x)$ and exploration effort $a(t,x)$ for several intermediate $t$'s.
At each $t$, production rate $q(t,x)$ is increasing in reserves level $x$ (asymptotically reaching $p(t) - \kappa_1$ as $ x\to\infty$),
while exploration effort $a(t,x)$ is decreasing in $x$, becoming zero $a(t,x) = 0$ for $x \ge 80$.
The monotonicity of $q(t,\cdot)$ and $a(t,\cdot)$ is due to decreasing marginal value of reserves,
which is consistent with the results in \cite{LS-Cournot,LudkovskiYang14}.
Both production and exploration rates decrease in $t$,
because the discovery rate $\lambda(t)$ is decreasing,
which gives decreasing motivation for exploration and in turn lowers production as marginal value of reserves rises. The above $q(t,x)$ and $a(t,x)$ for $0\leq t \leq T$ and $0 \leq x \leq X_{max}$ will be used in the next Section~\ref{sec: Numerical method for transport equation} as input to compute the evolution of reserves distribution.

\begin{figure}[htb]
\begin{center}
\begin{tabular}{cc}
\begin{minipage}{0.47\textwidth}
\includegraphics[width=0.98\textwidth, height=2.4in,trim=1in 2.65in 1in 2.2in]{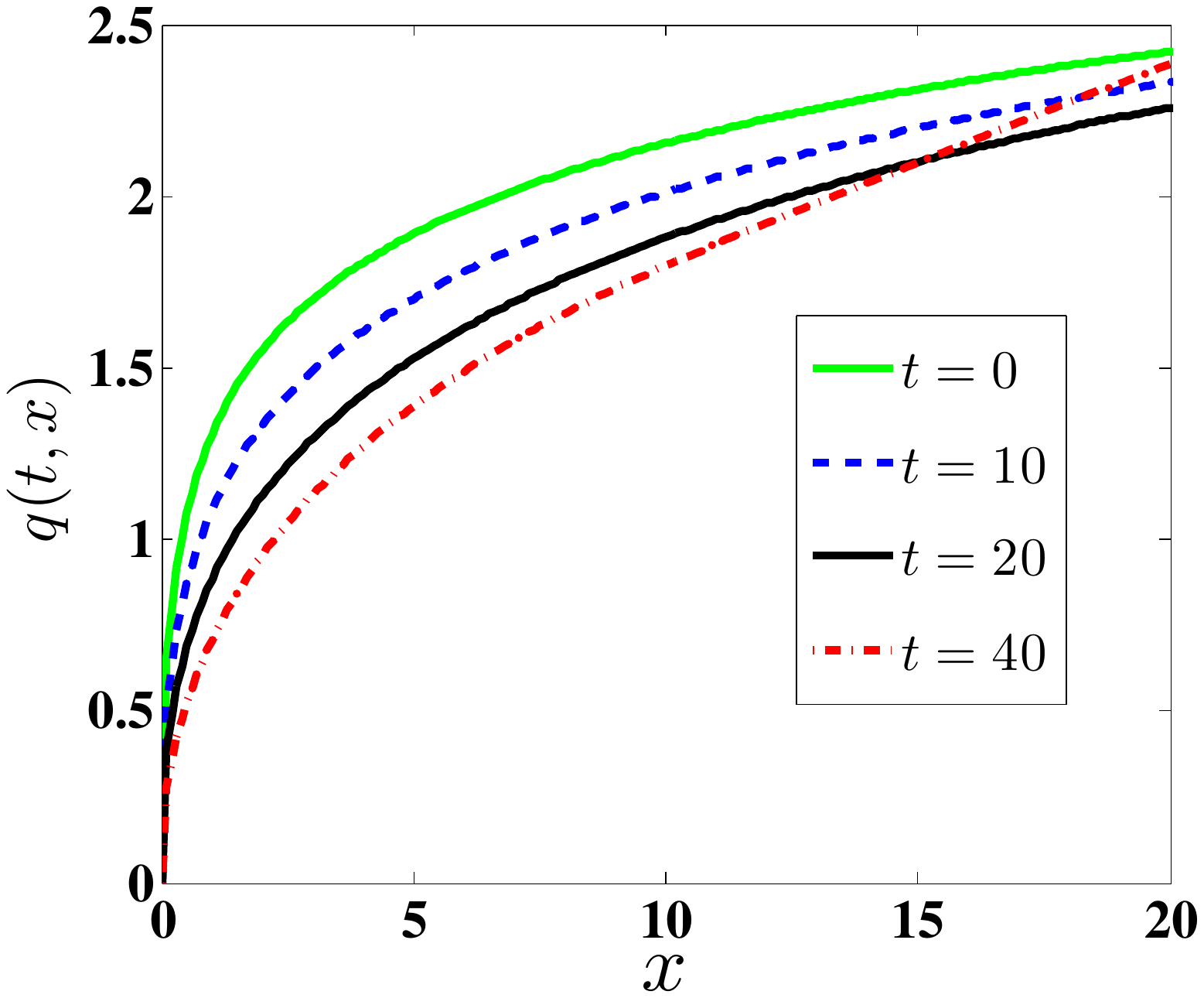}
\end{minipage} &
\begin{minipage}{0.47\textwidth}
\includegraphics[width=0.98\textwidth, height=2.4in,trim=1in 2.65in 1in 2.2in]{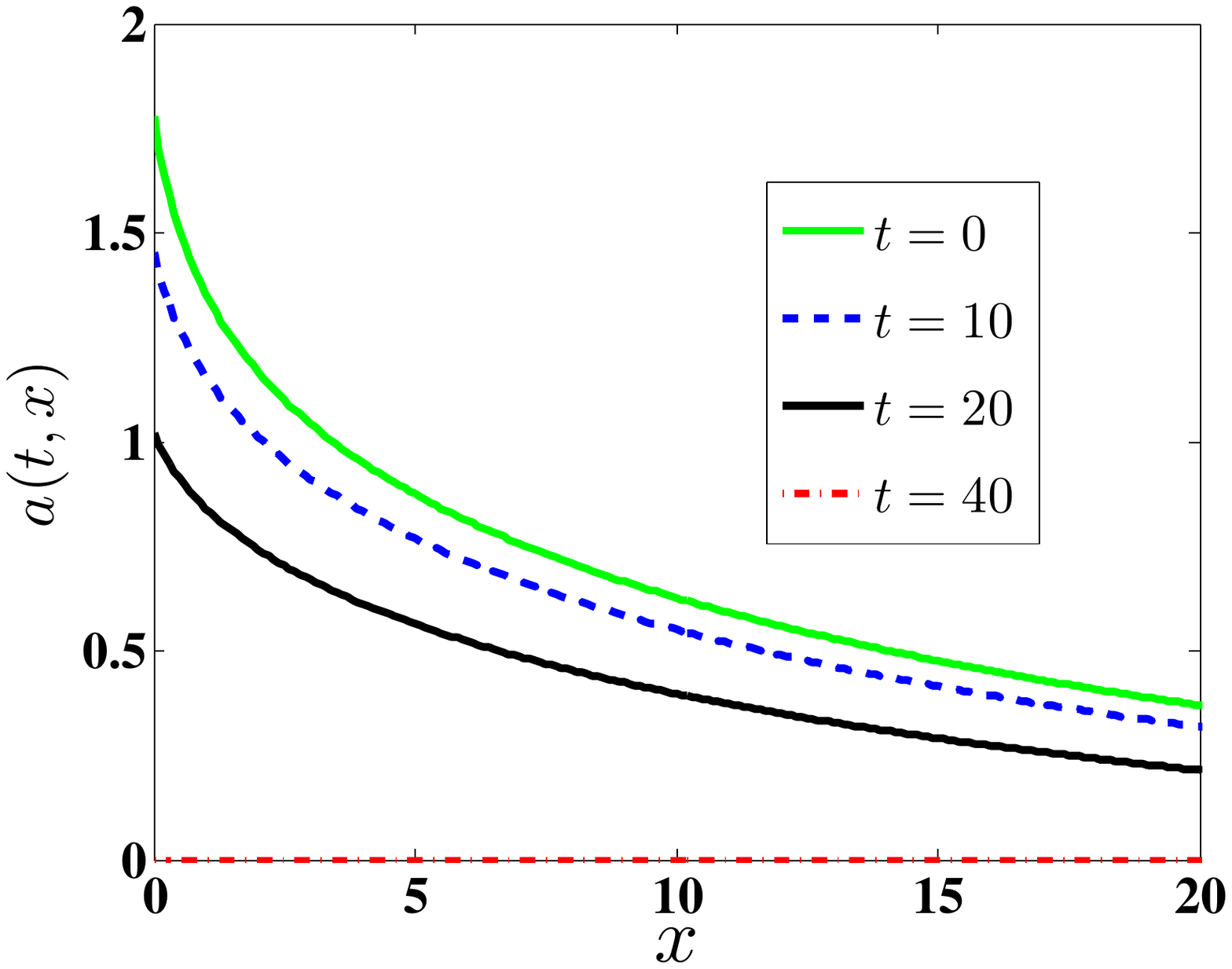}
\end{minipage}
\end{tabular}
\begin{minipage}{0.97\textwidth}
\caption{ Production and exploration controls $(q, a)$ associated with the HJB equation \eqref{mfg HJB} under constant price $p(t)=3$ and $\lambda(t) = (1-0.025t)^+, 0\le t \le T$.
Left panel: production rate $q(t,x)$. Right panel: optimized exploration rate $a(t,x)$. }
\label{fig: q_single_agent and a_single_agent}
\end{minipage}
\end{center}
\end{figure}


\subsection{Numerical scheme for transport equation}
\label{sec: Numerical method for transport equation}

We now assume given controls $q(t,x), a(t,x)$ and take up the evolution of the reserves distribution.
To numerically solve the transport equations of $\eta(t, x)$ we use a fully explicit finite difference scheme which replaces derivatives with discretized increments of the respective functions over a grid. We use the same partition in the space domain $[0, X_{max}]$ using $\Delta x$ as in the previous section.
To justify this bounded domain for the $x$-variable, recall the discussion about saturation level $x_{sat}$ at the end of Section \ref{sec: Game value function of a representative player} which motivates us to assume that $a(t,x) = 0$ for $x$ large enough, and in turn implies that $\eta(t,x) =0$ for $x$ large enough (e.g.~$x \ge \sup_t x_{sat}(t) + \delta$, with the additional assumption that the support of the initial distribution $\eta_0$ is also bounded). Thus, we apply the numerical boundary condition $\eta(t, X_{max}) = 0$ for all $t$.  (Even if $a(t,x) > 0$ for all $x$, we still expect that the right tail of $\eta$ should become negligible for $x$ large and so can be numerically truncated at $X_{Max}$.) Furthermore, we partition the time domain $[0, T]$ using a mesh $0 = t_0 < t_1 < ... < t_N = T$ with $t_n = n \Delta t$. To handle the boundary at $x=0$, the values $\eta(t, \cdot)$ and $q(t, \cdot)$ at $x=0+$ are numerically approximated by $\eta(t, x_1)$ and $q(t, x_1)$, respectively.

With the above setup, we
approximate both derivatives in time and in space by forward difference quotients:
\begin{equation}
\frac{\partial}{\partial t} \eta(t_n, x_m)
 \approx
 \frac{\eta(t_{n+1}, x_m) - \eta(t_n, x_m)}{\Delta t} ,
\quad
\frac{\partial}{\partial x} \eta(t_n, x_m)
 \approx
 \frac{\eta(t_{n}, x_{m+1}) - \eta(t_n, x_{m})}{\Delta x} .
\label{approximation to space derivative}
\nonumber
\end{equation}
By choosing $d = \lfloor \frac{\delta}{\Delta x} \rfloor$,  so that $x_m -\delta \simeq x_{m-d}$
we approximate the integral term in \eqref{eq:transport-2}--\eqref{eq:transport-3} with a Riemann sum
\begin{equation}
 -\int^{x_m}_{(x_{m}-\delta)_+} \lambda(t) a(t,x)  \eta(t, dx)
\approx
 \sum_{j=m-d + 1 \vee 1}^{m} \lambda(t_n) a(t_n, x_j) \left( \eta(t_n, x_{j-1}) - \eta(t_n, x_{j}) \right),
\label{approximation to integral term}
\end{equation}
where $\eta(t_n, x_{j-1}) - \eta(t_n, x_{j})$ is the proportion of producers with reserves in the interval $[x_{j-1}, x_j]$.

We start with given initial condition
$\eta(t_0, x_m) = \eta_0(x_m)$, $m=0, \ldots, M$, and solve forward in time using the right-edge boundary condition $\eta(t_n, x_M) = 0$, $n=0, ..., N$.
We take $\eta(t_n,x_0) = 1$ and interpret $\eta(t_n,x_1) \approx \eta(t_n, 0+)$ so that $\pi(t_n) = \eta(t_n, x_0) - \eta(t_n, x_1)$. 
We then solve for $\eta(t_{n+1}, \cdot)$ forward in space, splitting into cases according to $x_m \lessgtr \delta$. 
For $0<x_m<\delta$ (i.e.~$m=1,2,\ldots$),  which corresponds to \eqref{eq:transport-2}, 
we obtain the numerical value of $\eta(t_{n+1}, x_m)$ as
\begin{align}
\eta(t_{n+1}, x_m)
& = \eta(t_{n}, x_m)
 + \Delta t q(t_n, x_m) \frac{\eta(t_n, x_{m+1}) - \eta(t_n, x_{m})}{\Delta x} \notag\\
& \quad
- \Delta t \sum_{j=1}^m \lambda(t_n) a(t_n, x_j) \left( \eta(t_n, x_j) - \eta(t_n, x_{j-1})  \right) .
\label{numerical eta 1}
\end{align}
where the term for $j=1$ corresponds to $\lambda(t_n) a(t_n,0) \pi(t_n)$ in \eqref{eq:transport-2}.
For $x_M > x_m > \delta$, cf.~-\eqref{eq:transport-3},
 we obtain the numerical value of $\eta(t_{n+1}, x_m)$ by
\begin{align}
\notag
\eta(t_{n+1}, x_m)
& =  \eta(t_n, x_m)
- \Delta t \sum_{j=m-d+1}^{m} \lambda(t_n) a(t_n, x_j) \left( \eta(t_n, x_{j} ) - \eta(t_n, x_{j - 1}) \right) 	\notag	\\
 & \qquad + q(t_n, x_m) \left( \eta(t_n, x_{m+1}) - \eta(t_n, x_{m}) \right) \frac{ \Delta t }{\Delta x} .
\label{numerical eta 2}
\end{align}

Note that the above equations require only the values $a(t_n, x_m), q(t_n, x_m)$ and there is no difficulty in combining a method-of-lines approach for the HJB portion of the MFG equations with the above fully discretized finite-difference scheme for the transport equation.

\subsubsection{Illustrating the Evolution of Reserves Distribution}\label{sec:transport-ex}
As an example suppose that the initial reserves distribution has a parabolic initial density $m_0(x)$
$$
m_0(x)=\frac{6x(u-x)}{u^3} \quad \Leftrightarrow \quad \eta_0(x) = 1- 3 (x/u)^2 + 2 (x/u)^3   \quad\text{for}\quad 0 \le x \le u,
$$
and $m_0(x) = 0$ otherwise.
In the example shown in Figure \ref{fig: P_t and m_t example}, we take $u=10$, cf.~$m(0,x)$ on the left panel of the  Figure.
The evolution of boundary probability $\pi(t) = \PP(X_t = 0)$ and the density of reserves distribution $m(t,x) = - \frac{\partial}{\partial x}\eta(t, x)$ are shown in Figure~\ref{fig: P_t and m_t example}.
Numerically the density function is approximated by a difference quotient $m(t_n, x_m) \approx \frac{\eta(t_n, x_{m}) - \eta(t_n, x_{m+1})}{\Delta x}$. Since discovery rate $\lambda(t)$ decreases in time, the reserves density $m(t,x)$ shifts towards zero as time evolves, as shown on the left panel of Figure~\ref{fig:  P_t and m_t example}. Similarly,
the proportion $\pi(t)$ of producers with no remaining reserves increases in $t$ and zero global reserves are left shortly after discovery becomes impossible $\inf\{t : \pi(t) = 1 \} \simeq 41$, cf.~right panel of Figure \ref{fig:  P_t and m_t example}. 
We also note the discontinuity of $m(t,\cdot)$ at $x=\delta$ due to the discrete reserves jumps from $X_t = 0$.

\begin{figure}[htb]
\begin{center}
\begin{tabular}{rl}
\begin{minipage}{0.47\textwidth}
\includegraphics[width=0.98\textwidth,height=2.1in,trim=0.95in 3in 1.25in 2.95in]{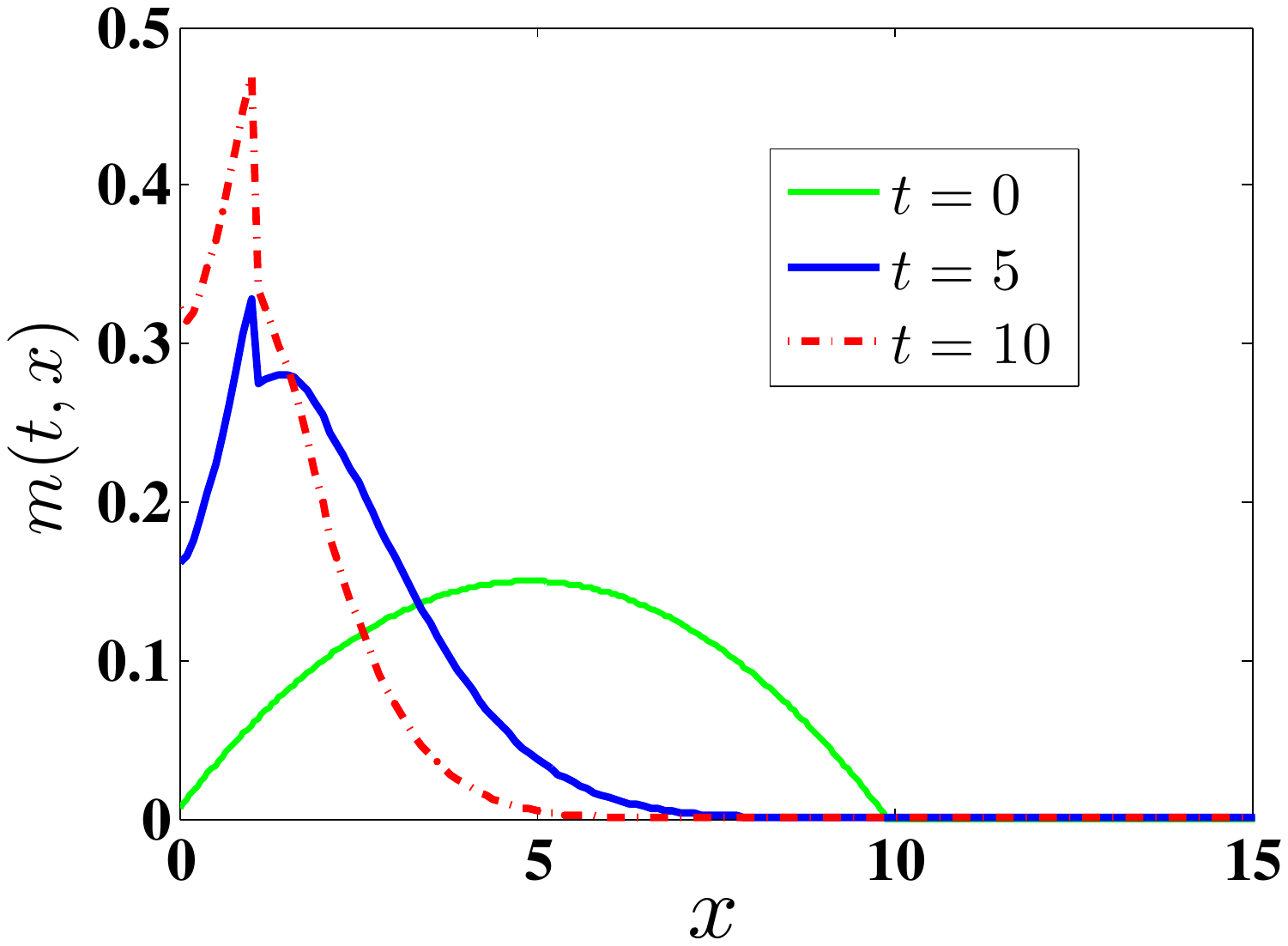}
\end{minipage} &
\begin{minipage}{0.47\textwidth}
\includegraphics[width=0.98\textwidth, height=2.1in,trim=.95in 3in 1.25in 2.95in]{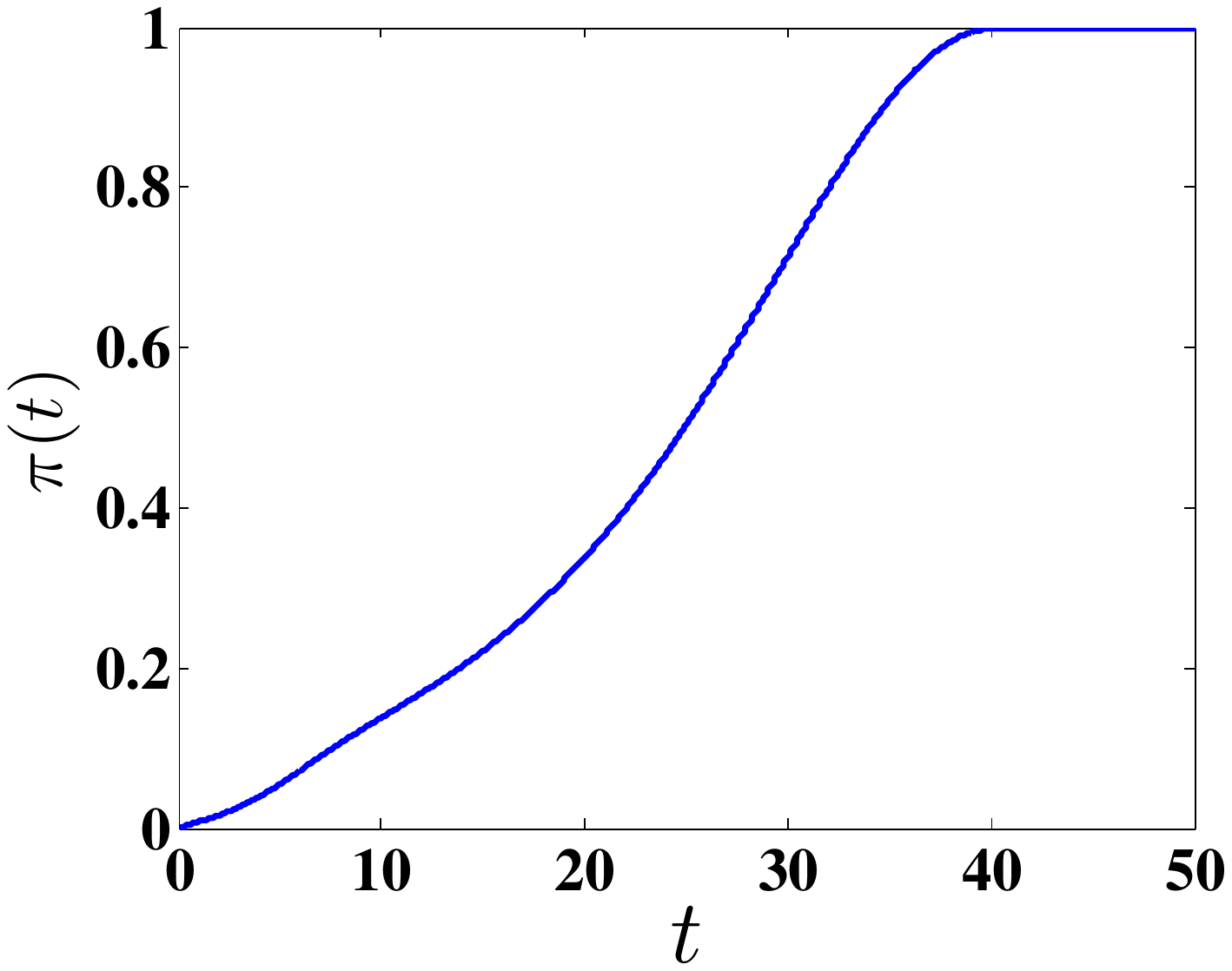}
\end{minipage}
\end{tabular}
\begin{minipage}{0.97\textwidth}
\caption{
Evolution of reserves distribution under the production and exploration controls $(q, a)$ obtained in Section \ref{sec: Numerical method for HJB equation}.
The discovery rate is $\lambda(t) = (1-0.025t)^+$ and unit amount of a discovery is $\delta = 1$.
Left panel:  Density of reserves distribution $m(t,x)=-\frac{\partial}{\partial x} \eta(t, x)$ for several $t$'s.
Right:  Proportion of producers with no reserves $\pi(t) = \PP(X_t = 0)$. After $t=41$ all reserves are exhausted.   }
\label{fig:  P_t and m_t example}
\end{minipage}
\end{center}
\end{figure}


\subsection{Numerical scheme for the MFG system}
\label{sec: Numerical method for the system of HJB and transport equations}

We introduce an iterative scheme to solve the system of coupled HJB and transport equations.
Our solution strategy consists of a loop over the following three steps. The loops are repeated over the iterations $k=0,1,\ldots$ until numerical convergence.

To initialize, we start with an initial price process $p^{(0)}(t)$ (greater than $\kappa_1$, to ensure strictly positive production rate). In Step 1, given the current  $p^{(k)}(\cdot)$, the numerical scheme in Section~\ref{sec: Numerical method for HJB equation} is implemented for the HJB equation, outputting the optimal production $q^{(k)}$ and exploration $a^{(k)}$ rates.
 Next in Step 2, these $q^{(k)}$ and $a^{(k)}$ are substituted into the transport equation to solve for $\eta^{(k)}$, following the scheme in Section~\ref{sec: Numerical method for transport equation}.
We then compute the total production $Q^{(k)}$ by using a Riemann sum to approximate the integral of $q^{(k)}(t,x)$ with respect to $\eta^{(k)}(t,\cdot)$.  Finally, in Step 3 we update the price to $p^{(k+1)}$.
Observe that if $p^{(k)}(t)$ is lower than equilibrium price $p^\ast(t)$ for all $t\in [0, T]$,
the resulting $Q^{(k)}(t)$ will be lower than the equilibrium $Q(t)$. As a result, $D^{(-1)}(Q^{(k)}(t))$ will be higher than $p^\ast(t)$, and vice versa. Thus to speed up convergence, we take $p^{(k+1)}(t)$ in the next iteration to be the average of $p^{(k)}(t)$ and $D^{-1}(Q^{(k)}(t))$. Numerically we observe that this yields a monotone sequence of $p^{(k)}(t)$'s, improving convergence to the equilibrium  $p^\ast(t)$.

{\bf Step 0}.  Start with an initial guess $p^{(0)}(t), t \in [0,T]$ of market price.

{\bf Step 1}. For iteration $k=0,1,2,...$, and given $p^{(k)}(\cdot)$, solve the HJB equation \eqref{mfg HJB equation} to obtain $v^{(k)}(t,x)$
and  the corresponding $q^{(k)}(t,x)$ and $a^{(k)}(t,x)$ as in \eqref{eq:q-star-2}-\eqref{eq:a-star-2}.

{\bf Step 2}. With the above $q^{(k)}$ and $a^{(k)}$ solve the transport equation to obtain $\eta^{(k)}(t,x)$ satisfying \eqref{transport equation}.

{\bf Step 3}. Update the market price via  the new total quantity of production
\begin{equation*}
p^{(k+1)}(t):= \frac{ D^{-1}\left( Q^{(k)}(t)\right) + p^{(k)}(t) }{2}
\ \quad\text{with}
\quad Q^{(k)}(t) = \sum_{m=1}^{M-1} q^{(k)}(t,x_m) [ \eta^{(k)}(t,x_m) - \eta^{(k)}(t,x_{m+1})].
\end{equation*}

{\bf Repeat} Steps 1 --- 3 until convergence in the sup-norm defined as
$\left\| \cdot  \right\|_\infty :=\sup_{[0, T]\times[0, X_{max}]}|\cdot|$.
Iteration will stop when tolerance of error $TolError$ is satisfied
\begin{align}
\left\| v^{(k+1)} - v^{(k)}  \right\|_\infty < TolError, 	\quad\text{and}\quad
\left\| \eta^{(k+1)} - \eta^{(k)}  \right\|_\infty < TolError .
\label{eq:tol-error}
\end{align}

We continue with the running example where the discovery rate is
$ \lambda(t) = \left(1- 0.025t\right)^+$, $\delta = 1$ and initial price process is $p^{(0)}(t) = 3 \forall t$. Recall that the solutions obtained in Sections~\ref{sec:hjb-ex} and \ref{sec:transport-ex} can be viewed as the first iteration $k=0$ of the above scheme.
Figure~\ref{fig: Convergence of iterations} illustrates the iterations over $k=0,1,\ldots$ with the resulting HJB value functions $v^{(k)}(t,x)$  at fixed time $t=10$.
 In each iteration $k$, if $v^{(k)}(t, x)$ is lower than the equilibrium value $v^\ast(t, x)$ for all
$x \in [0, X_{max}]$ with some $t$ fixed, then in the next iteration $v^{(k+1)}(t, x)$ will move up towards the equilibrium level $v(t, x)$. This pointwise monotone convergence in $x$ is observed in Figure~\ref{fig: Convergence of iterations}. Numerical convergence with a tolerance of $TolError=10^{-6}$ in \eqref{eq:tol-error} is achieved  after $k=4$ iterations.

\begin{figure}[htb]
\begin{center}
\includegraphics[width=0.5\textwidth,height=2.2in,trim=0.85in 3in 1.1in 2.65in]{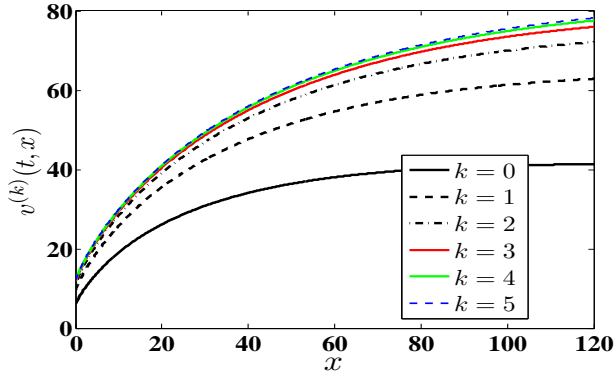}
\begin{minipage}{0.97\textwidth}
\caption{Convergence of the numerical scheme in Section \ref{sec: Numerical method for the system of HJB and transport equations}. We start with initial guess $p^{(0)}(t) = 3 \forall t \in [0, T]$, and discovery rate $\lambda(t) = (1-0.025t)^+$.
Game value function at $t=10$, $v^{(k)}(t, x)$, converges after $k\geq 4$ iterations.    }
\label{fig: Convergence of iterations}
\end{minipage}
\end{center}
\end{figure}

Figure~\ref{fig: Q A R exhaustible} shows the resulting evolution of total production $Q(t)$, total discovery rate $A(t)$, and total reserves level $R(t)$.
Total reserves $R(t)$ decrease as production proceeds; in turn
decreasing $R(t)$ lowers the total production rate $Q(t)$ and raises market price $p(t)$. Interestingly we observed a hump shape in $t \mapsto A(t)$: initially exploration efforts rise, then peak and gradually decline. This complex relationship is driven by the changing exploration success parameter $\lambda(t)$ (that discourages exploration as time progresses) and the reserves distribution $\eta(t,x)$ (which encourages exploration as reserves tend to get depleted on average).

To get some further insights, we compare these results with the non-exploration (NE) case that has zero discovery rate $\lambda(t)=0$. When $\lambda(t)=0$, no exploration effort will be made $a^\ast(t,x) \equiv 0$ as there is no hope to have any discovery. Consequently, producers simply gradually extract their initial reserves, eventually leading to total depletion, $R^{NE}(t) = 0$  for $t > 10.5$ in the Figure. This postponement of the reserves ``Doomsday'' is illustrated in the right-most panel of Figure~\ref{fig: Q A R exhaustible} that plots the evolution of the proportion of exhausted producers' $\pi(t)$.
In comparison, for a model with exploration ultimate depletion only happens around $t=41$ (recall that $\lambda(t) = 0$ after $t=40)$. In fact at $t=10$, less than 10\% of producers have no reserves. As expected, because exploration increases global reserves, $R^E(t) \ge R^{NE}(t)$, the respective  marginal value of reserves is lower and hence production is boosted, $Q^E(t) \ge Q^{NE}(t) \forall t$. Thus, exploration not only delays exhaustion but also unambiguously raises revenues.

\begin{figure}[htb]
\begin{center} 
\begin{tabular}{rccl}
\begin{minipage}{0.23\textwidth}
\includegraphics[width=0.97\textwidth,height=2in,trim=1.25in 2.95in 1.35in 2.85in]{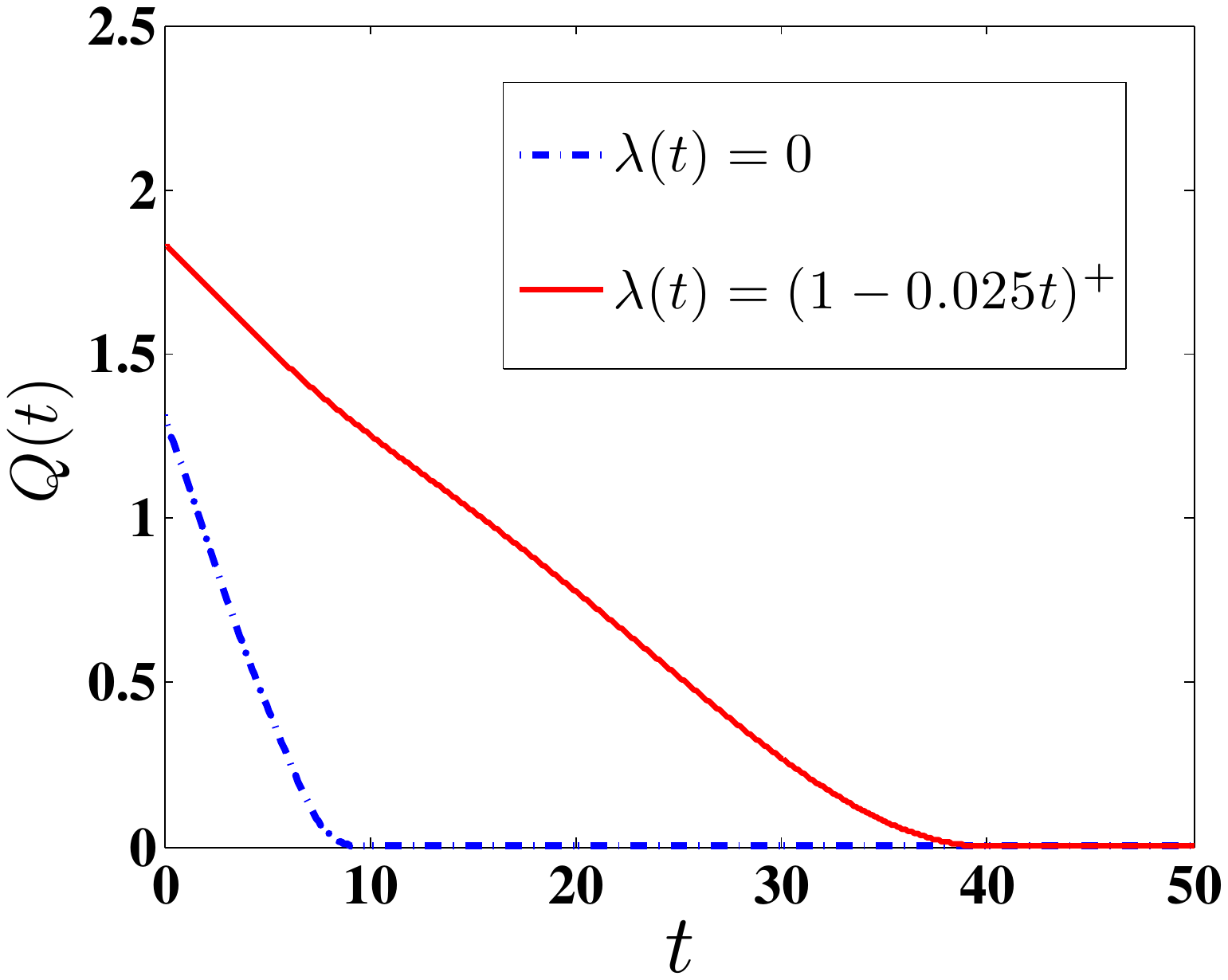}
\end{minipage} &
\begin{minipage}{0.23\textwidth}
\includegraphics[width=0.97\textwidth,height=2in,trim=1.35in 2.95in 1.35in 2.85in]{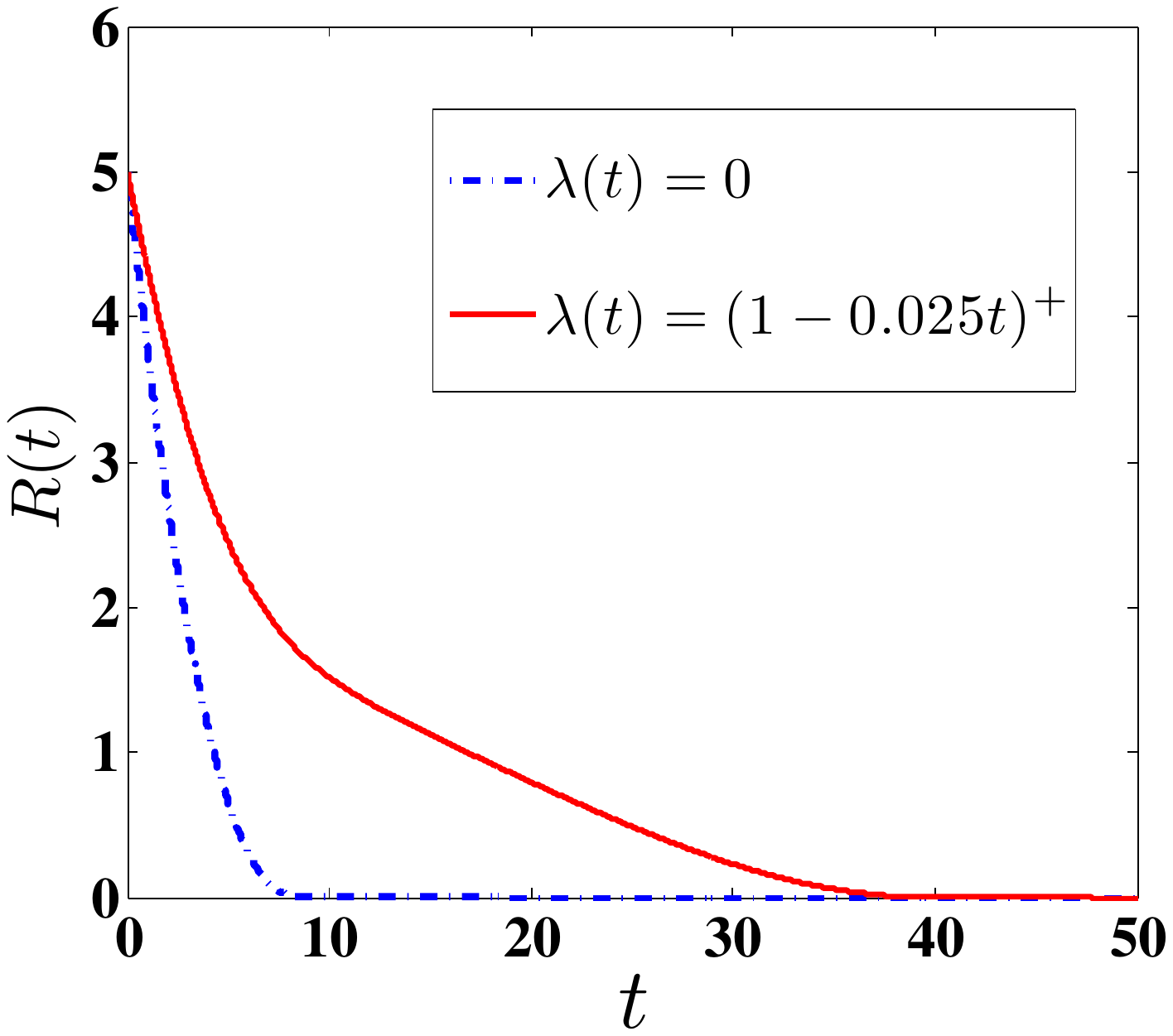}
\end{minipage} &
\begin{minipage}{0.23\textwidth}
\includegraphics[width=0.97\textwidth,height=2in,trim=1.25in 2.95in 1.35in 2.85in]{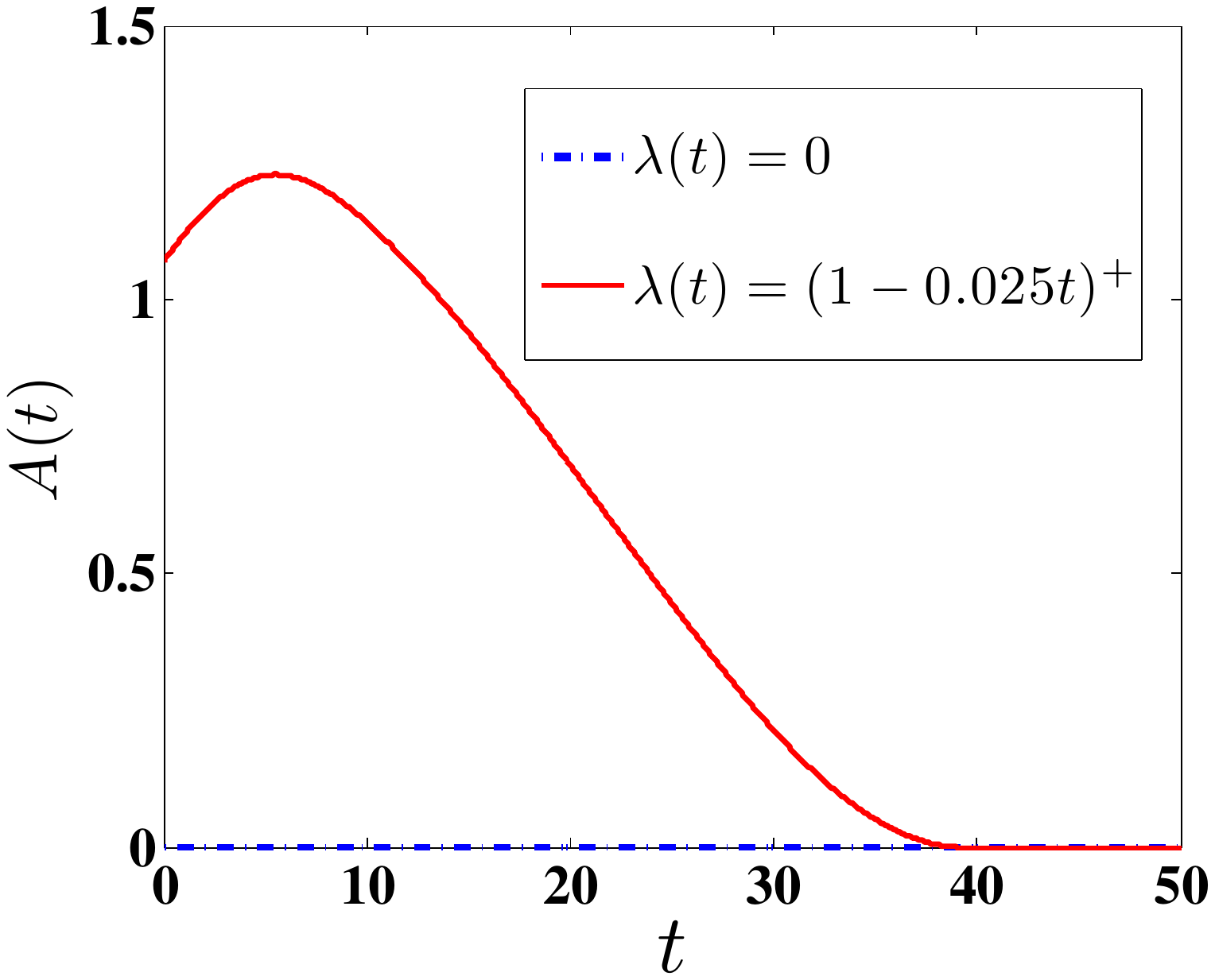}
\end{minipage} &
\begin{minipage}{0.23\textwidth}
\includegraphics[width=0.97\textwidth,height=2in,trim=1.25in 2.95in 1in 2.85in]{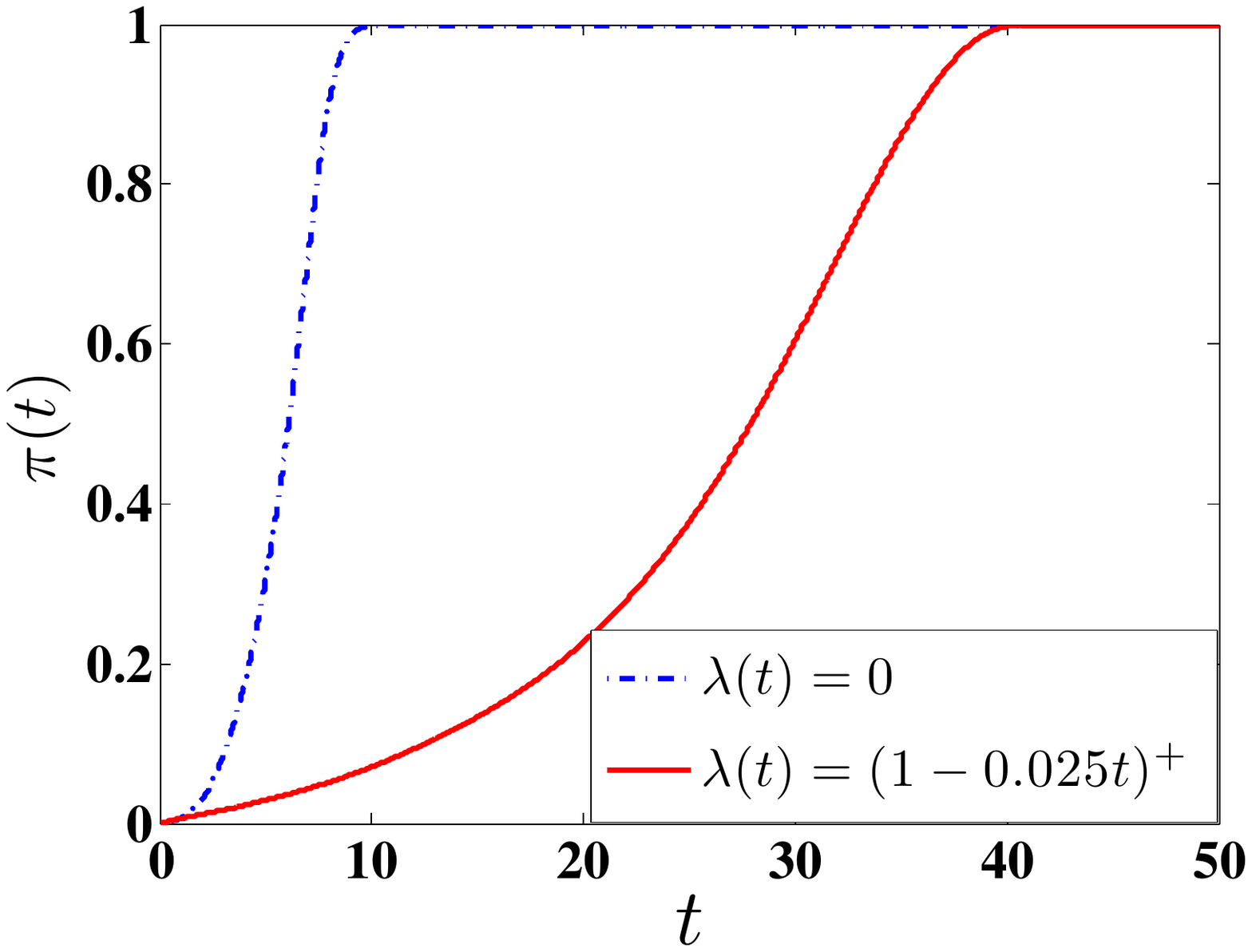}
\end{minipage}
\end{tabular}
\end{center}
\begin{minipage}{0.97\textwidth}
\caption{From left to right: evolution of total production $Q(t)$, total reserves $R(t)$, total discovery rate $A(t)$ and proportion $\pi(t)$ of producers with no reserves as a function of $t$. We show
 discovery rate $\lambda(t) = (1-0.025t)^+$, in comparison to no-exploration $\lambda(t)=0$. 
\label{fig: Q A R exhaustible}}
\end{minipage}
\end{figure}

\section{Stationary mean field game Nash equilibrium}
\label{sec: Stationary mean field game Nash equilibrium}

In Section~\ref{sec: Mean field game Nash equilibrium} we studied a generic model with time-inhomogeneous discovery rate $\lambda(t)$, which would typically be taken to be decreasing in time.
When there are still abundant resources underground, it is reasonable to assume that the discovery rate is time-homogeneous $\lambda(t) = \lambda$, for some $\lambda > 0$.
Thanks to exploration, the commodity used up for production can be compensated by new discoveries, and thus a \emph{stationary} level of production and exploration can be obtained. In this section, we discuss such stationary MFG equilibria denoted by
$(\tilde{q}, \tilde{a}, \tilde{\eta})$.
Specifically, if the reserves has initial distribution $X_0\sim \tilde{\eta}$, and all the players apply the strategy $ q_t = \tilde{q}(x;  \tilde{\eta})$ and $a_t = \tilde{a}(x; \tilde{\eta})$, then the reserves process
\begin{align}\label{eq:stat-X}
d X_t  = - \tilde{q}(X_t)\mathds{1}_{\{ X_t >0\}} dt + \delta d \tilde{N}_t
\end{align}
has the distribution $\tilde{\eta}(\cdot)$ for all $t>0$, that is, the reserves distribution is invariant in time.
We define the stationary objective functional $\tilde{\mathcal{J}}$ of a player with  current reserves level $x$ and conditionally on a reserves distribution $\tilde{\eta}(\cdot)$ as
\begin{equation}
\label{stationary mfg objective functional}
\tilde{\mathcal{J}}( \tilde{q}, \tilde{a} ;x, \tilde{\eta})
: =  \E
\left\{ \int_0^{\infty}\left[ D^{-1}\left( \tilde{Q}(\tilde{\eta}) \right) \tilde{q}(X_t) - C_q(\tilde{q}(X_t)) - C_a(\tilde{a}(X_t)) \right] e^{-rt} dt \ \middle| \ X_0 = x\right\} ,
\end{equation}
where $\tilde{Q}(\tilde{\eta}) := - \int_0^\infty \tilde{q}(x) \tilde{\eta}(dx) $ is the stationary aggregate production.

\begin{defn}[Stationary MFG MNE]
\label{defn: Stationary mean field game Nash equilibrium}
A stationary mean field game Nash equilibrium is a triple
$\left( \tilde{q}^\ast, \tilde{a}^\ast, \tilde{\eta} \right)$ such that for $(X_t)$ from \eqref{eq:stat-X}
the distribution of reserves $\tilde{\eta} = \PP( X_t \ge x) \forall t$ is unchanged under the strategies
$(\tilde{q}^\ast, \tilde{a}^\ast)$,
and
\begin{equation}
\label{stationary mfg optimality condition}
\tilde{v}( x) \equiv \tilde{\mathcal{J}}( \tilde{q}^\ast, \tilde{a}^\ast; \tilde{\eta})
\geq
\tilde{\mathcal{J}}(q, a; \tilde{\eta}),
\quad
\forall (q, a)\in \mathcal{A} .
\end{equation}

\end{defn}

The following Proposition \ref{prop: Stationary mean field game partial differential equations} gives the system of stationary HJB and transport equations for $\tilde{v}, \tilde{\eta}$ under a constant discovery rate $\lambda > 0$. Intuitively, it is equivalent to the equations in the previous section after dropping the dependence on $t$. Consequently, we pass from PDE's to ordinary differential equations in $x$.

\begin{prop}[Characterizing stationary MFG equilibrium]
\label{prop: Stationary mean field game partial differential equations}

The stationary value function $\tilde{v}$ and upper-CDF $\tilde{\eta}$ satisfy:
\begin{align}
\label{stationary mfg HJB equation}
  r \tilde{v}( x) & =
   \left[ - C_a( \tilde{a}^\ast( x) ) + \tilde{a}^\ast( x)\lambda \Delta_x \tilde{v}( x) \right]
+ \left[ \tilde{p}  \tilde{q}^\ast( x) - C_q(\tilde{q}^\ast( x))
- \tilde{q}^\ast( x) \tilde{v}'( x) \right] , \quad
  x > 0 ; \\
\label{stationary mfg transport equation}
& \begin{cases}
0  = \lambda \tilde{a}^\ast (0) (1-\tilde{\eta}(0+))
- \int_{0+}^x \lambda \tilde{a}^\ast (z) \tilde{\eta}(dz)
+ \tilde{q}^\ast(x ) \tilde{\eta}'(x) , \qquad 0<x \leq \delta ,   \\
0 = - \int_{x -\delta}^x \lambda \tilde{a}^\ast( z )
\tilde{\eta}( dz)
+ \tilde{q}^\ast(x) \tilde{\eta}'(x), \qquad\qquad\qquad x > \delta  ,
\end{cases}
\end{align}
where the equilibrium stationary production and exploration rates ($\tilde{q}^\ast$, $\tilde{a}^\ast$) and  price $\tilde{p}$ are
\begin{align}
 & \tilde{q}^\ast(x)
 = \frac{1}{\beta_1} \left( L  -  \tilde{Q} - \kappa_1
-  \tilde{v}'(x) \right)^+ ,  	\notag	\\
& \tilde{a}^\ast( x)
 =  \frac{1}{\beta_2}\left( \lambda \Delta_x \tilde{v}( x) - \kappa_2 \right)^+ ,
\notag \\
& \tilde{p} = D^{-1}\left( \tilde{Q}  \right) = L + \int_0^\infty \tilde{q}^\ast(x) \tilde{\eta}(dx) ,
\end{align}
with $\tilde{Q}$ uniquely determined by the equation
\begin{align}
\tilde{Q}
& =
- \int_0^\infty
\frac{1}{\beta_1}\left( L - \kappa_1 -  \tilde{v}'(x) - \tilde{Q} \right)^+
\tilde{\eta}(dx) .  \label{eq:stat-Q}
\end{align}

\end{prop}

Similar to \cite{LS-Cournot},
the boundary condition for $\tilde{v}(0)$ is determined by
\begin{equation}
\tilde{v}(0)
= \sup_{a \geq 0} \E \left[ e^{-r\tau} \tilde{v}(\delta) -
 \int_0^\tau e^{-rt} C_a(a) dt \right]
= \sup_{a \geq 0} \frac{a\lambda \tilde{v}(\delta) - C_a(a)}{r+a\lambda} .
\label{stationary HJB boundary}
\end{equation}

\begin{remark}
If the rate of new discoveries is zero, $\lambda=0$ then from the transport equation \eqref{stationary mfg transport equation} we have $\tilde{\eta}'(x) = 0$
for all $x > 0$, which implies that there is no producer with positive reserves level in the long run.
\end{remark}

\subsection{Solving for stationary MFG equilibria}
\label{sec: Examples of stationary mean field game}

For the stationary MFG developed in \eqref{stationary mfg HJB equation}-\eqref{stationary mfg transport equation}, the iterative scheme introduced in section
\ref{sec: Numerical method for the system of HJB and transport equations}
is not directly applicable. The  challenge lies in solving the stationary transport equation
\eqref{stationary mfg transport equation}. We see that the singularity at $x=0$ creates in effect
a \emph{free boundary} at $x=0$ that describes the balance between the density for $x> 0$ and the point
mass $\tilde{\pi}$ of exhausted producers. It is not clear how to directly handle this free boundary without ending with an intractable global system of coupled nonlinear equations.

To overcome this issue, we exploit the link between the time-dependent and time-stationary MFG models. Specifically, with a constant discovery rate $\lambda$ and large horizon $T$, the strategies $(v^\ast, a^\ast)$ have only weak dependence on $t$. Thus, we expect a convergence of the reserves process $X_t$ to an invariant distribution, since with a feedback, time-independent control it forms a recurrent Markov process on $\mathbb{R}_+$. This suggests to take $T$ large, solve the MFG on $[0,T]$ and then ``extract'' a $(\tilde{v}, \tilde{a}, \tilde{\eta})$ to approximate the true stationary solution.

A numerical illustration of a non-stationary MFG with constant discovery rate
$\lambda(t) = \lambda \equiv 1$ is shown in Figure \ref{fig: stationary mfg}.
The lower panels of Figure \ref{fig: stationary mfg} show the evolution of total production $Q(t)$, total discovery $A(t)$, and total reserves level $R(t)$, which are defined by
\eqref{total production}-\eqref{total discovery}.
We observe a boundary layer for small $t$ (roughly $t \in [0, 12]$) arising from the non-equilibrium initial distribution $\eta_0(dx)$, and another boundary layer (roughly for $t \in [45,50]$) arising from the terminal condition $v(T,x) = 0$. The latter causes $\lim_{t \to T} R(t) = 0, \lim_{t \to T} A(t) = 0$ observed in the plots. (Note that as the horizon is approached, total production rises in order to spend down all reserves and reach $R(T) = 0$.) At the same time, for the intermediate $t$'s all the quantities are effectively time-independent and hence should be close to the stationary MFG equilibrium solution. In particular, due to the conservation of reserves, $R(t) \simeq 1.9$ for $t\in [15,45]$ we observe that $Q(t) \simeq A(t)$ on that time interval. Similarly, the respective reserves distribution $\eta(t,x)$ is almost independent of $t$, cf.~the plot of $\pi(t)$ in Figure~\ref{fig: stationary mfg}. Put another way, the actual value of the horizon $T$ is essentially irrelevant as it only determines where the end-of-the-world boundary layer appears (around $t=T-5$ in the plot) and has negligible effect on the solution prior to that.

A rigorous treatment of this phenomenon has been given in Cardaliaguet et al.~\cite{Cardaliaguet12} for a locally coupled MFG, and Cardaliaguet et al.~\cite{CardaliaguetPorretta13} for a special case of a non-local coupling. According to \cite{CardaliaguetPorretta13}, for each $t\in [0, T]$, the solution
$(v(t, x), \eta(t, x))$ of a non-stationary MFG model converges in $L^2$-norm to the solution
$(\tilde{v}(x), \tilde{\eta}(x))$ of stationary MFG model as $T\to \infty$. Furthermore, in their setting the difference between stationary and non-stationary mean field game equilibrium solutions, measured by $L^2$-norm, is minimized at $t = T/2$.
Extending these proofs to the setting of Cournot MFGs with exploration is left for future research.

In light of these results, we can obtain an approximate solution of the stationary MFG MNE by solving the non-stationary equations \eqref{mfg HJB equation} and \eqref{mfg transport equation} with constant discovery rate $\lambda(t) \equiv \lambda$,
employing the same iterative scheme as in Section~\ref{sec: Numerical method for the system of HJB and transport equations}. Then the solution $(v(t, x), \eta(t, x))$ at $t = T/2$ is taken as approximate solution of the stationary mean field game model
\eqref{stationary mfg HJB equation}-\eqref{stationary mfg transport equation},
i.e., $\tilde{v}(x) \approx v(T/2, x)$ and $\tilde{\eta}(x) \approx \eta(T/2, x)$
for all $x \in [0, X_{max}]$. A related approach was taken in Chan and Sircar~\cite{ChanSircar16} where the stationary MFG solution was obtained by solving non-stationary transport equation coupled with stationary HJB equation and taking the large time limit.

In the example shown in Figure \ref{fig: stationary mfg} we had $T=50$, and so we use the intermediate solution $( v(t, x), \eta(t, x) ) \approx (\tilde{v}(x), \tilde{\eta}(x))$ at $t = T/2 = 25$ as an approximation to the corresponding time-stationary MFG. The upper left panel of Figure~\ref{fig: stationary mfg} shows the (approximate) stationary reserve density $\tilde{m}(x)  \approx  \frac{\partial}{\partial x}\eta(25, x)$. We observe that $\tilde{m}(x)$ increases in $x$ for
$0<x \leq \delta$ where the rate of discovery is higher than the rate of production, and decreases for $x > \delta$.
Similarly, we can extract the stationary total production $\tilde{Q}$, total discovery $\tilde{A}$, and total reserves level $\tilde{R}$ by looking at $Q(t), A(t), R(t)$ at $t=T/2=25$. (Due to conservation of mass $\tilde{A} = \tilde{Q}$.)

\begin{figure}[htb]
\begin{center}
\begin{tabular}{rl}
\begin{minipage}{0.45\textwidth}
\includegraphics[width=0.98\textwidth, height=2in,trim=1in 2.85in 1.2in 2.3in]{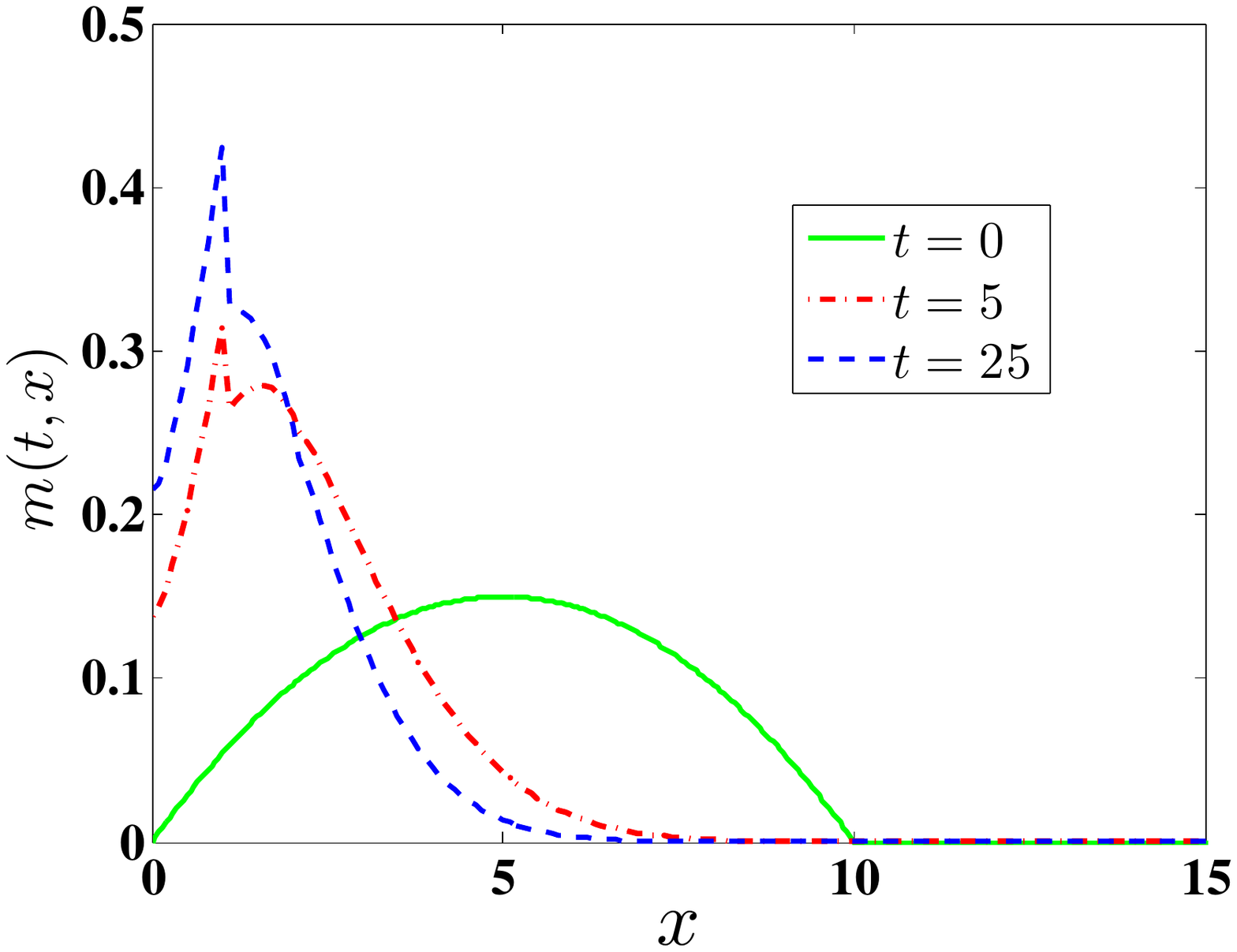}
\end{minipage} &
\begin{minipage}{0.45\textwidth}
\includegraphics[width=.98\textwidth, height=2in,trim=1.2in 2.85in 1.2in 2.3in]{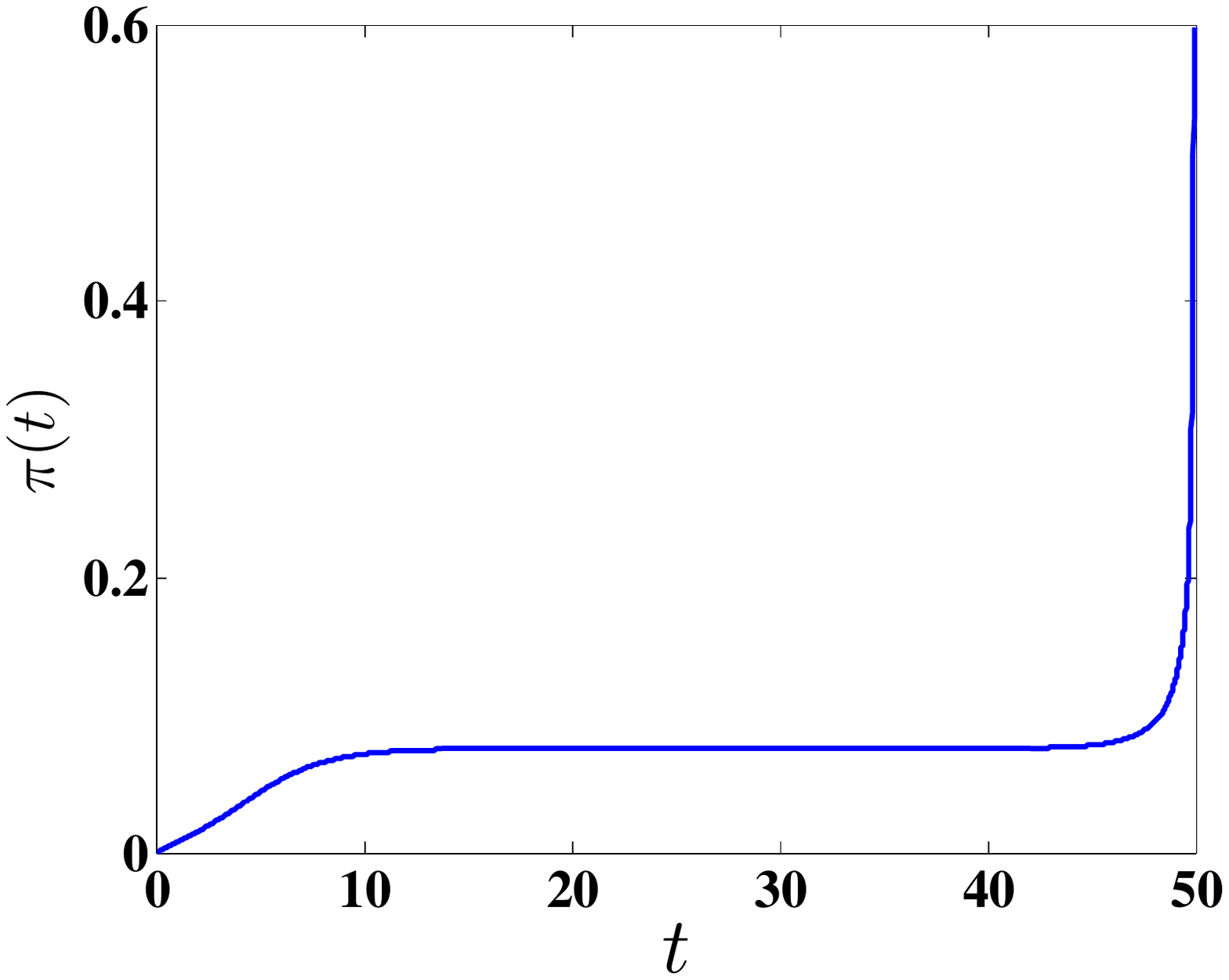}
\end{minipage} \\
\begin{minipage}{0.45\textwidth}
\includegraphics[width=0.98\textwidth, height=2in,trim=1in 2.85in 0.5in 2.3in]{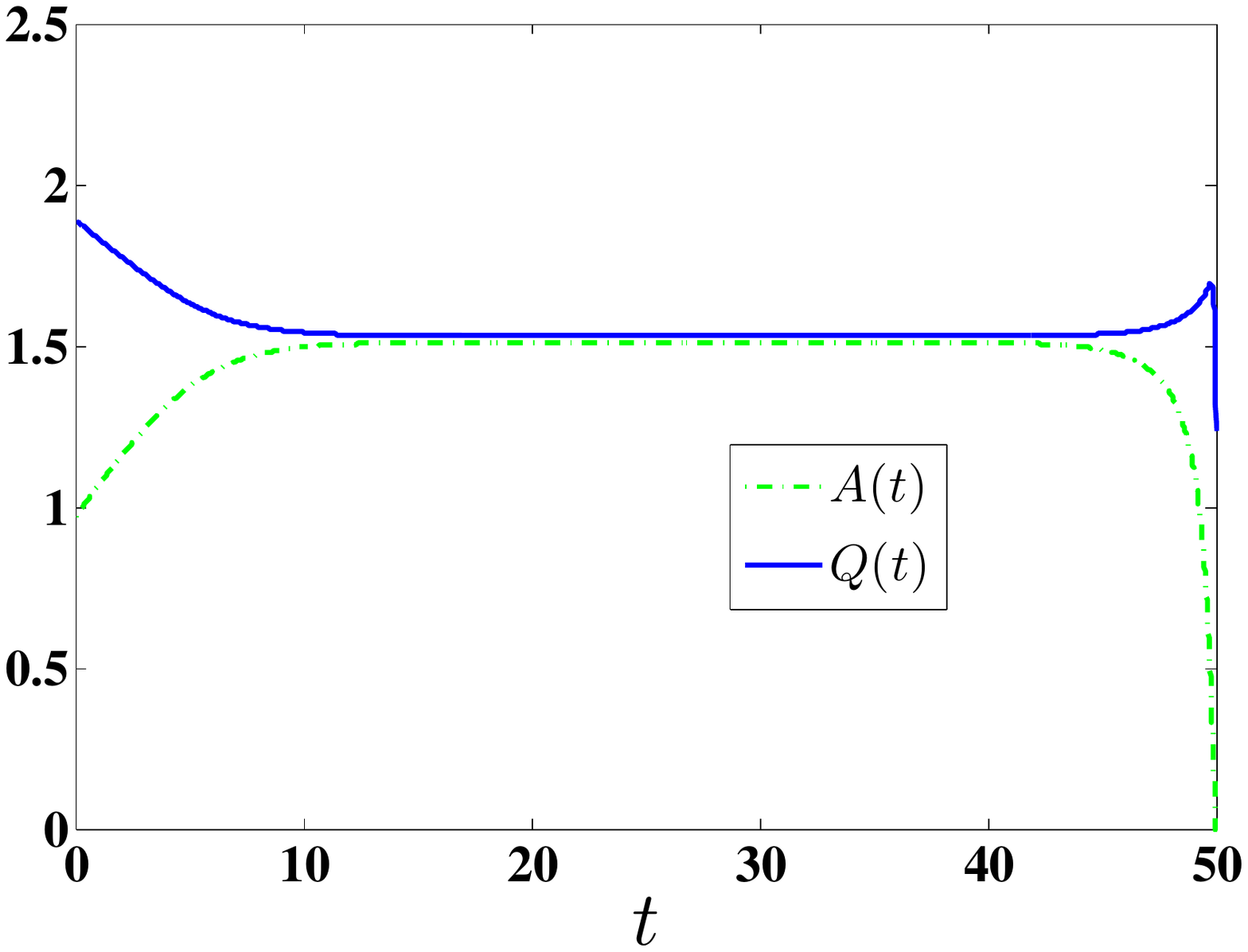}
\end{minipage} &
\begin{minipage}{0.45\textwidth}
\includegraphics[width=0.98\textwidth, height=2in,trim=1.2in 2.85in 1.2in 2.3in]{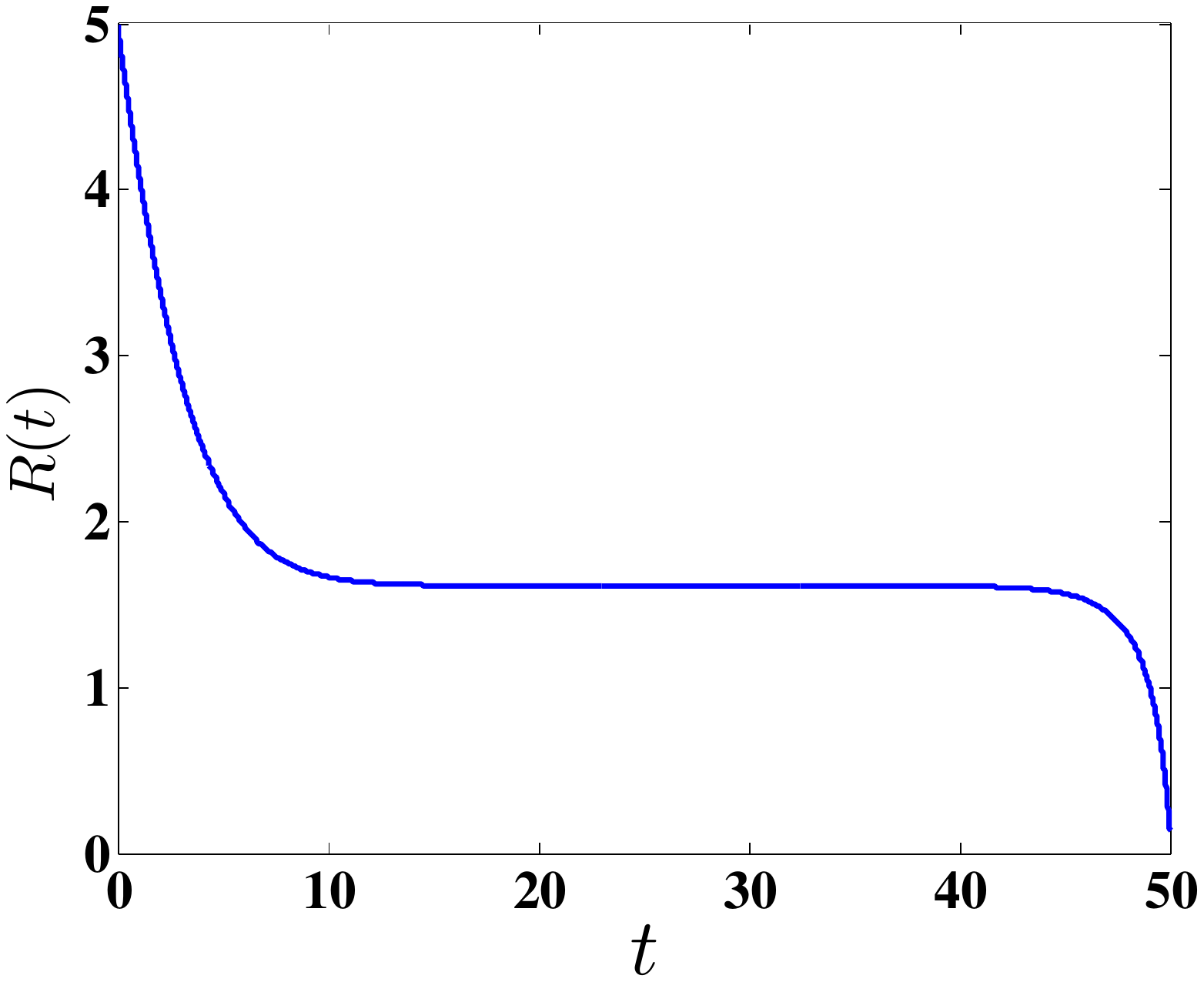}
\end{minipage}
\end{tabular}
\begin{minipage}{0.97\textwidth}
\caption{MFG solution with a constant $\lambda(t) \equiv \lambda = 1$ and $T=50$ to illustrate the relationship between the time-dependent and stationary solutions.
Upper left panel: Density $m(t, x)$ of reserves distribution.
Upper right: Proportion $\pi(t)$ of producers without reserves.
Lower left: Total exploration rate $A(t)$ and total production $Q(t)$.  Lower right: Total reserves $R(t)$. }
\label{fig: stationary mfg}
\end{minipage}
\end{center}
\end{figure}

\subsection{Comparative Statics for the Stationary MFG}
It is instructive to study the effect of exploration on the equilibrium of the stationary mean field game.
Figure~\ref{fig: stationary_Q_against_lambda stationary_R_against_lambda stationary_bdryP_against_lambda} shows the effect of discovery rate $\lambda$ on the aggregate stationary quantities $\tilde{Q}$, $\tilde{A}$,
and $\tilde{R}$, all of which have positive relation with $\lambda$.
As discoveries take place faster with larger $\lambda$, the marginal value of each discovery decreases which yields an ambiguous effect to exploration effort $\tilde{a}^\ast$. In the top right panel of Figure~\ref{fig: stationary_Q_against_lambda stationary_R_against_lambda stationary_bdryP_against_lambda}) we observe that for low values of $\lambda$, $\lambda \to \tilde{a}^\ast(x; \lambda)$ increases, i.e.~exploration is encouraged by higher likelihood of discovery. However, for high $\lambda$'s, $\lambda \to \tilde{a}^\ast(x; \lambda)$ is decreasing  pointwise as the producer becomes ``lazy'' and does not see a need to work as hard, since new reserves are so easy to come by. In aggregate across $x$, we do observe a positive relation between $\lambda$ and total discovery rate $\tilde{A}$ (top left panel). Due to $\tilde{A}=\tilde{Q}$, this translates into higher aggregate production and lower prices.

Easier discoveries also raise the stationary level of reserves $\tilde{R}$, although the underlying impact on
$\tilde{\eta}$ is non-monotone. This is illustrated in the
bottom right panel of Figure~\ref{fig: stationary_Q_against_lambda stationary_R_against_lambda stationary_bdryP_against_lambda} which plots the density $\tilde{m}(x)$ for several different $\lambda$'s and highlights multiple phenomena of interest. On the one hand, we observe that $\lambda \tilde{a}^\ast(0)$ monotonically increases in $\lambda$ which reduces the expected time until next discovery at $x=0$ and hence lowers the stationary proportion $\tilde{\pi}$. In the same vein, $\tilde{R}$ rises in $\lambda$ and shifts 
 the whole $\tilde{m}$ to the right. On the other hand,, the spread, i.e.~variance of $\tilde{m}$ starts falling as $\lambda$ keeps rising. Thus, for low $\lambda$, $\tilde{\eta}$ is more spread out and $\tilde{\pi}$ is higher; for high $\lambda$ $\tilde{\eta}$ is concentrated around the average $\tilde{R}$. Moreover, the support of $\tilde{m}$ has a hump shape in $\lambda$. Recall that due to exploration saturation, $\tilde{m}$ is supported on $[0, \tilde{x}_{sat}+\delta]$ where $\tilde{x}_{sat}$ is the saturation level. We find that $\tilde{x}_{sat}$ first rises and then falls in  terms of $\lambda$. For example, when $\lambda=1$, we have $\tilde{x}_sat = 64.8$ which can be compared to $\tilde{x}_{sat}(\lambda=0.2) = 60.7$ and $\tilde{x}_{sat}(\lambda = 10) = 23.9$. In the latter situation when $\lambda$ is very big, there is no reason to hold many reserves (instead resources can be replenished almost instantaneously), so $\tilde{v}(x)$ approaches its horizontal asymptote quickly and hence exploration only takes place for small $x$.
A further phenomenon is that
when $\lambda$ is very small, e.g.~$\lambda < 0.05$ in Figure
\ref{fig: stationary_Q_against_lambda stationary_R_against_lambda stationary_bdryP_against_lambda},
exploration stops entirely ($\tilde{A}=0$ leading to $\tilde{R}=0$) in stationary equilibrium.
This occurs because when $\kappa_2>0$ and $\lambda$ is small enough,
the expected addition of value $\lambda \Delta_x \tilde{v}(x)$ is always smaller than the cost $\kappa_2$ and thus no exploration efforts will be made. Thus, when discoveries are ``too difficult'', exploration will cease
even if there are still potential new reserves remaining underground, $\lambda > 0$.


\begin{figure}[htb]
\begin{center}
\begin{tabular}{cc}
\begin{minipage}{0.3\textwidth}
\includegraphics[width=0.98\textwidth,height=1.9in,trim=1.2in 3in 1.25in 2.7in]{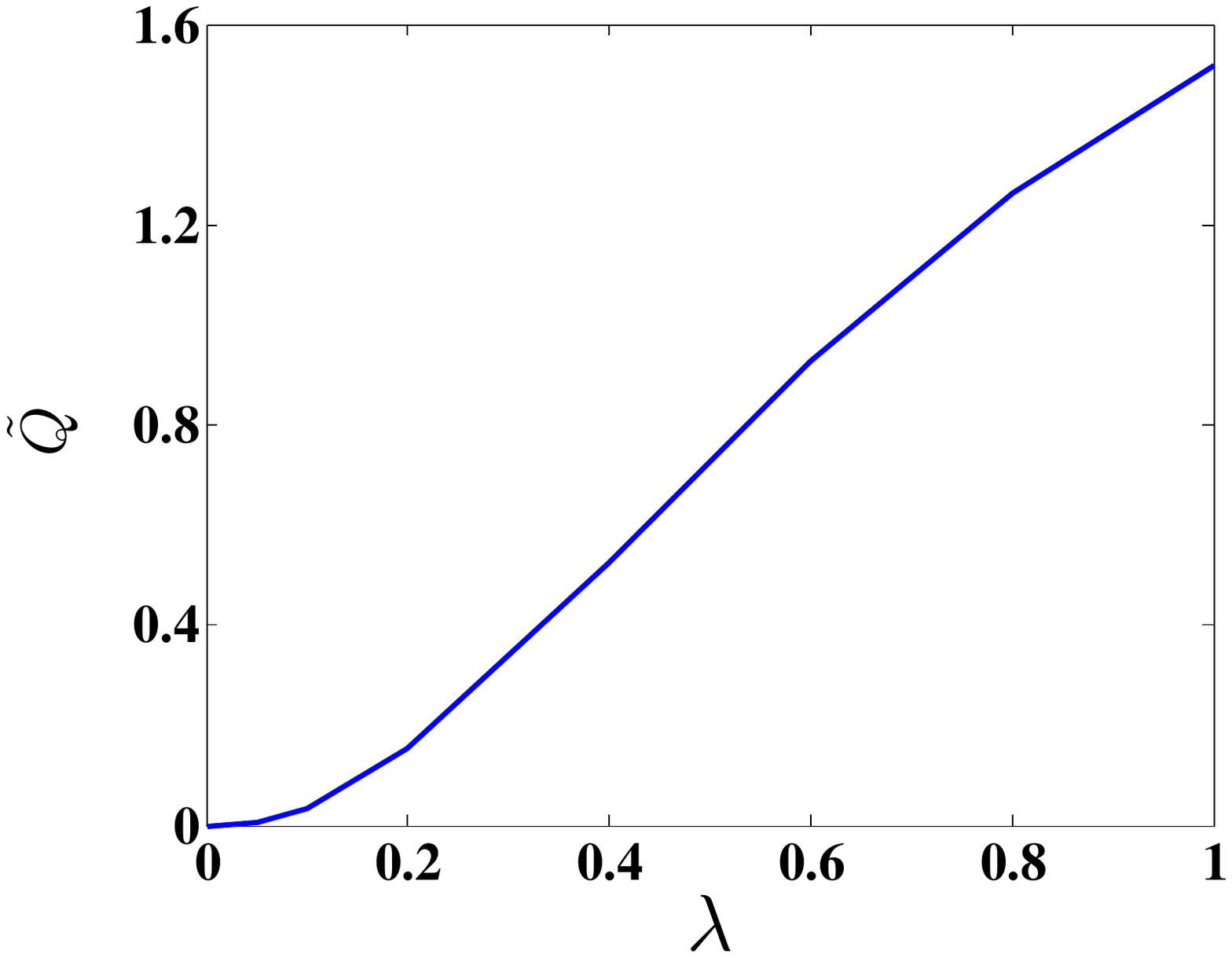}
\end{minipage} &
\begin{minipage}{0.5\textwidth}
\includegraphics[width=0.98\textwidth,height=1.9in,trim=0.8in 2.95in 1.2in 2.8in]{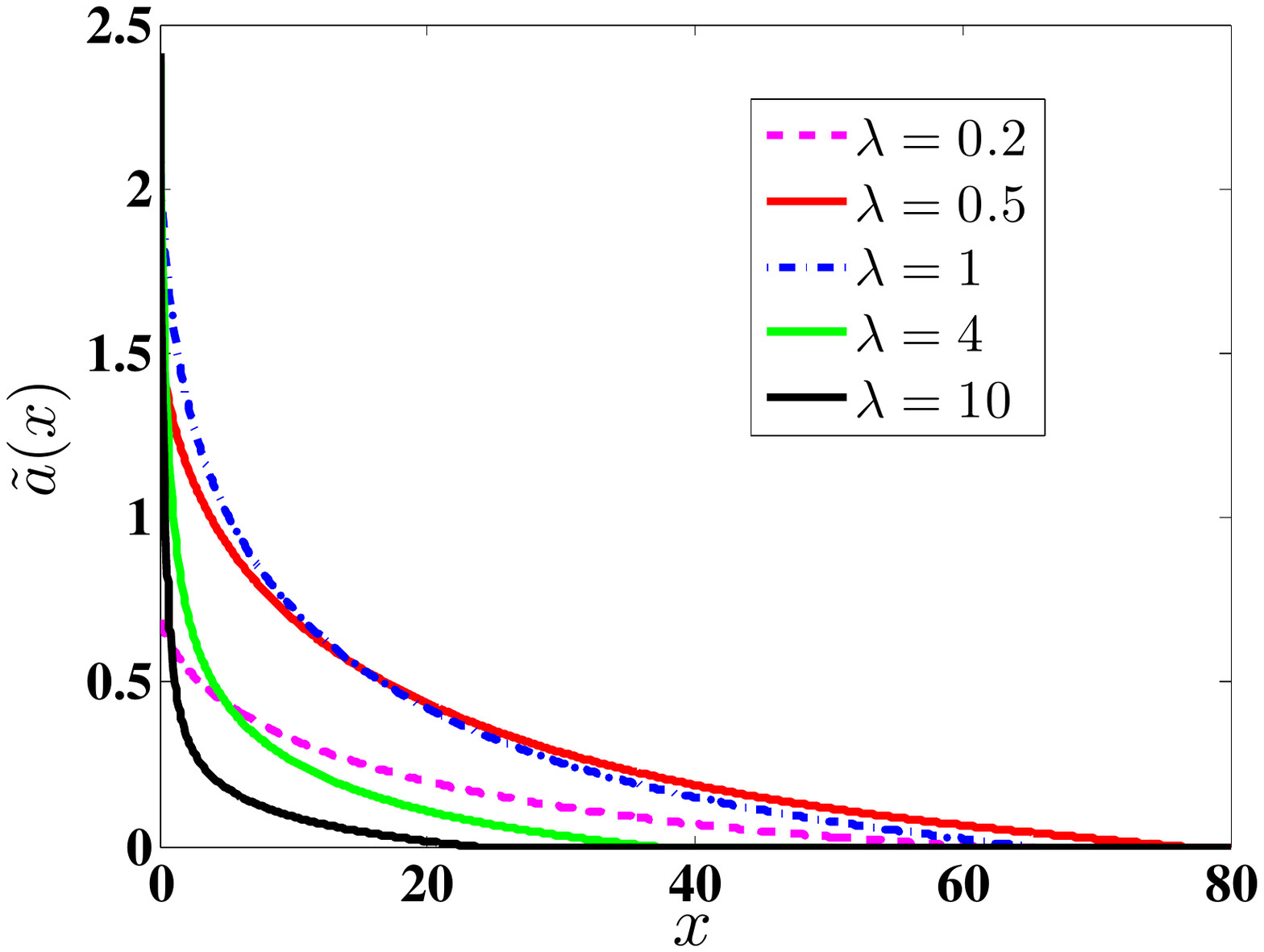}
\end{minipage} \\
\begin{minipage}{0.3\textwidth}
\includegraphics[width=0.98\textwidth,height=1.9in,trim=1.2in 2.95in 1.2in 2.65in]{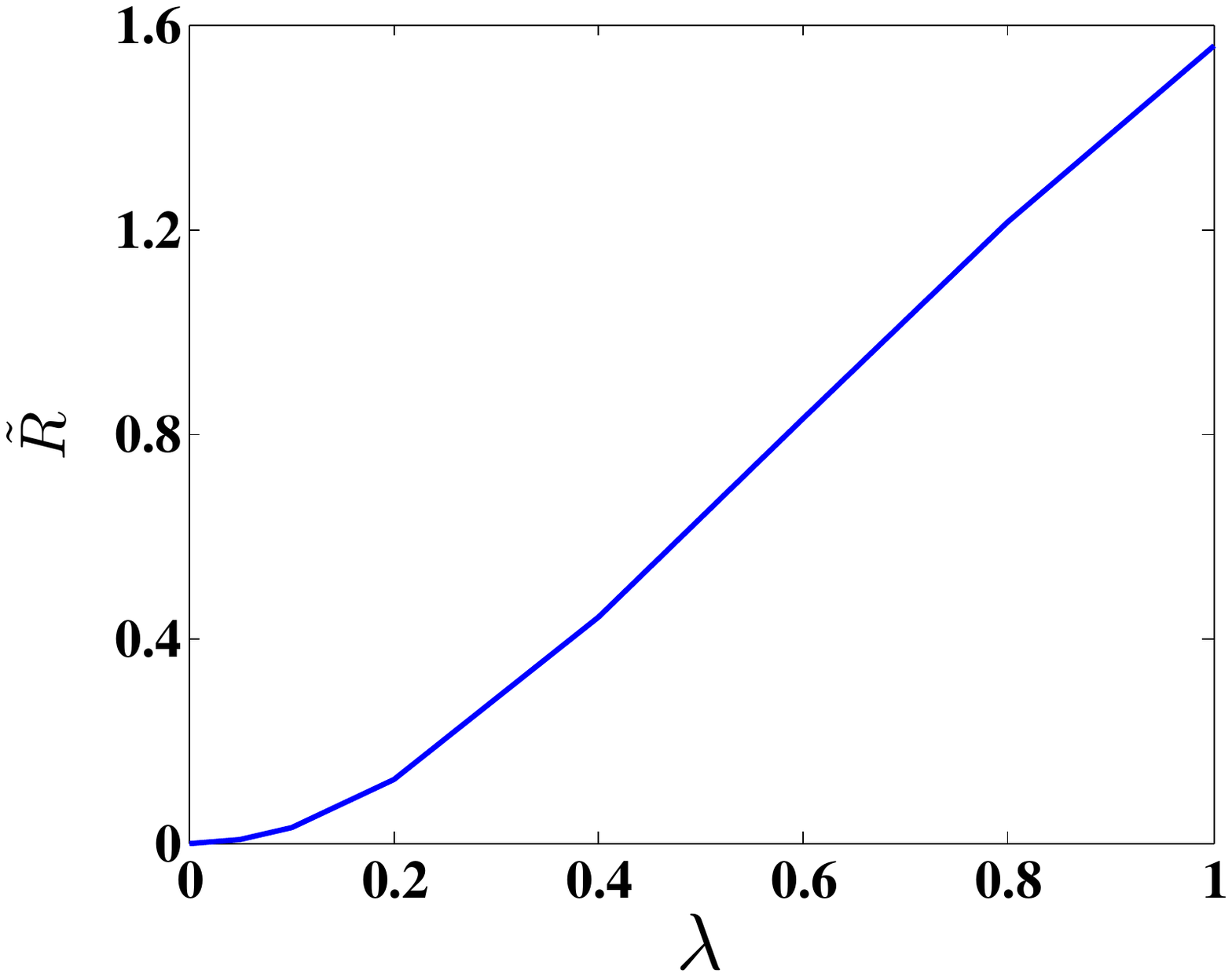}
\end{minipage}
 & 
\begin{minipage}{0.5\textwidth}
\includegraphics[width=0.98\textwidth,height=1.9in,trim=1.2in 3in 1.42in 2.85in]{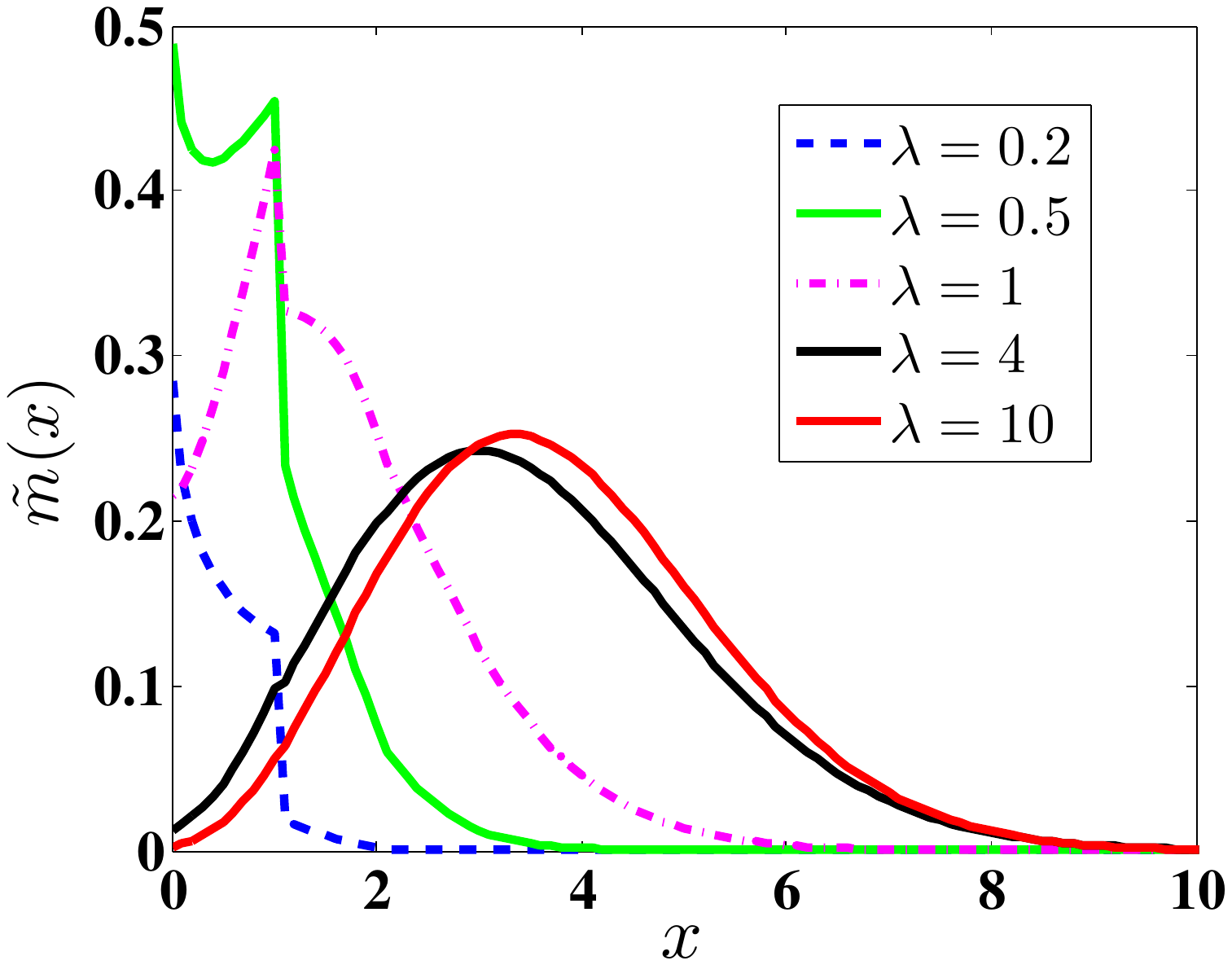} 
\end{minipage}
\end{tabular}
\begin{minipage}{0.97\textwidth}
\caption{ Stationary MFG solution as a function
 of  discovery rate $\lambda$.
\emph{Top Left} panel: Stationary aggregate exploration/production $\tilde{A}=\tilde{Q}$.
\emph{Top Right:} Stationary exploration effort $\tilde{a}^\ast(x)$.
\emph{Bottom Left:} Stationary aggregate reserves $\tilde{R} = \int_0^\infty x \tilde{m}(dx) $;
\emph{Bottom Right:} Stationary distribution $\tilde{m}(x)$. Note that the total mass on $(0,\infty)$ is $1-\tilde{\pi}$ which depends on $\lambda$. As before, there is a discontinuity at $x=\delta =1$.
\label{fig: stationary_Q_against_lambda stationary_R_against_lambda stationary_bdryP_against_lambda}}
\end{minipage}
\end{center}
\end{figure}


\section{Fluid limit of exploration process}
\label{sec: Fluid limit of exploration process}

The stochasticity of the exploration process depends on two factors: the discovery rate $\lambda$ per unit exploration effort, and the size $\delta$ of each discovery.  To study the effect of randomness of the exploration process on equilibrium production and reserves distribution we introduce an asymptotic parameter $\epsilon>0$ (cf.~\cite{HaganCaflisch94}),
rescaling
$\lambda_\epsilon := \lambda / \epsilon$ and  $\delta_\epsilon := \delta \epsilon$.
As $\epsilon \downarrow 0$, we have the discovery rate $\lambda_\epsilon \uparrow \infty$ and unit discovery amount $\delta_\epsilon \downarrow 0$,
which means that the exploration process becomes more deterministic. In the sequel we use $\epsilon$ to index the respective MFG equilibria.

For the limiting case $\epsilon = 0$ the exploration process is fully deterministic. This is known as the fluid limit since we fully average out the stochasticity in $(X_t)$ without modifying its average (in the sense of expected value) behavior. Intuitively in the fluid limit, the difference term $\Delta_x v(t,x) = v(t,x+\delta)-v(t,x)$ becomes $\frac{\partial}{\partial x} v_0(t,x)$  and the integral becomes $ \delta a_0^*(t,x) \frac{\partial}{\partial x} \eta_0(t,x)$, removing the non-local term. The resulting MFG equations are given by \eqref{fluid limit HJB non stationary}--\eqref{mfg transport equation fluid limit non stationary} below. %
%
\begin{align}
\label{fluid limit HJB non stationary}
& 0  = \frac{\partial}{\partial t} v_0 (t , x)  - r v_0(t, x)
+ \frac{1}{2\beta_1} \left[ \left(p_0(t) - \kappa_1  - \frac{\partial}{\partial x} v_0(t , x) \right)^+ \right]^2 	\notag 	\\
& \quad + \frac{1}{2\beta_2} \left[  \left( \lambda \delta \frac{\partial}{\partial x}v_0(t,x) - \kappa_2  \right)^+  \right]^2; \\
 & 
\frac{\partial}{\partial t} \eta_0(t,x)
= \left(  - \lambda \delta a^\ast_0(t,x)
+ q^\ast_0(t,x) \right)
\frac{\partial}{\partial x} \eta_0 (t,x), \quad x > 0   ,
\label{mfg transport equation fluid limit non stationary}
\end{align}
where the optimal production rate $q^\ast_0$ and exploration rate $a^\ast_0$ are
\begin{align}
q^\ast_0(t, x)
&= \arg \max_{q \geq 0} \left[ p_0(t)q - C_q(q) - q \frac{\partial}{\partial x} v_0(t,x) \right]	\notag	\\
& = \frac{1}{\beta_1} \left( p_0(t) - \kappa_1 - \frac{\partial}{\partial x} v_0(t,x) \right)^+ ,
\label{optimal production fluid limit non stationary}
	\\
a^\ast_0(t, x)
& =\arg \max_{ a \geq 0} \left[ - C_a(a) +a \lambda \delta \frac{\partial}{\partial x} v_0(t,x) \right]
 =  \frac{1}{\beta_2}\left( \lambda \delta \frac{\partial}{\partial x} v_0(t,x) - \kappa_2 \right)^+,
\label{optimal exploration fluid limit non stationary}
\end{align}
and $p_0(t) = L + \int_0^\infty q^\ast_0(t,x) \eta_0(t,dx)$.
Note that there is no ``boundary'' at $x=0$ for $\eta_0$ because depletion is never explicitly encountered; it only imposes the constraint $a^\ast_0(t,0) \ge q^\ast_0(t,0)$.
The boundary conditions $v_0(t, 0)$ and $\frac{\partial}{\partial x} v_0(t, 0)$ are given explicitly by the following Lemma proven in Appendix~\ref{proof of lemma boundary conditions fluid limit non stationary HJB}.
\begin{lemma}
\label{lemma: boundary conditions fluid limit non stationary HJB}
The boundary conditions $v_0(t, 0)$ and $\frac{\partial}{\partial x} v_0(t, 0)$ satisfy
\begin{align}
& \frac{\partial}{\partial x} v_0(t, 0)
= \frac{\beta_2 (p_0(t) - \kappa_1 ) + \beta_1 \lambda \delta \kappa_2}{\beta_1 \lambda^2 \delta^2 + \beta_2} ;
\label{v0 partial x boundary condition}  \\
&v_0(t, 0) =
\int_t^T  \left[  \frac{\lambda \delta (p_0(s) - \kappa_1) - \kappa_2}{\beta_1 \lambda^2 \delta^2 + \beta_2}  \right]^2 (1+\lambda^2 \delta^2) \e^{-r(s - t)} ds .
\label{v0 boundary condition}
\end{align}

\end{lemma}

By taking $T \to \infty$, we may then consider the stationary fluid limit MFG. Mathematically, this yields the simplest setup as it removes the non-local ``delay'' term associated with discrete exploration, as well as the time-dependence, leaving with a coupled system of two ODE's. In fact, the following Proposition~\ref{prop: stationary mean field game equilibrium in fluid limit} implies that economically the stationary MFG in the fluid limit reduces to just a couple of algebraic relations.

\begin{prop}[Stationary mean field game equilibrium in fluid limit]
\label{prop: stationary mean field game equilibrium in fluid limit}
The stationary MFG MNE in fluid limit ($\epsilon = 0$) is summarized as
\begin{enumerate}
\item[(i).] The stationary reserves distribution is $\tilde{\pi}_0 =1$, i.e.~all producers hold no reserves, $\tilde{R}_0 = 0$.
\item[(ii).] The equilibrium total production $\tilde{Q}_0$ and market price in the fluid limit are given by
\begin{align}
\tilde{Q}_0 = \tilde{q}^\ast_0(0)
 =
\frac{[ (L-\kappa_1)\lambda \delta - \kappa_2 ]^+ }{\beta_2 + (1+\beta_1)\lambda \delta}, \quad\text{and} \qquad \tilde{p}_0 = L- \tilde{Q}_0;
\label{total production in fluid limit}
\end{align}

\item[(iii).] The equilibrium exploration control is $\tilde{a}^\ast_0(x) = 0$ $\forall x > 0$ and
\begin{equation} \label{equation: fluid limit boundary production and exploration relation}
    \tilde{a}^\ast_0(0) = \frac{1}{\delta \lambda} \tilde{q}_0^\ast(0).
    \end{equation}
\end{enumerate}

\end{prop}

The proof of Proposition~\ref{prop: stationary mean field game equilibrium in fluid limit} is Appendix \ref{section: proof of stationary mean field game equilibrium in fluid limit}.
In the case of fluid limit $\epsilon=0$, discovery of new resources happens in a completely deterministic way, thus it is not necessary to hold reserves for production.
Producers starting with positive reserves will not explore until reserves run out.
Once reserves level reaches zero,  equation \eqref{equation: fluid limit boundary production and exploration relation} implies that a player without reserves will choose production and exploration strategies such that the production rate exactly equals the rate of reserves increment due to his exploration effort. This explains how zero reserves can be sustained in equilibrium. Overall, the above Proposition shows that the stationary equilibrium with deterministic exploration is trivial, i.e.~only $x=0$ matters and the system of ODE's effectively collapses to algebraic equations linking $\tilde{Q}_0$ and $\tilde{A}_0$ to model parameters. This shows that the stochastic model is strictly more complex than the deterministic one.

\subsection{Numerical scheme and illustration}
\label{sec: Numerical example of fluid limit model}


The iterative scheme in Section~\ref{sec: Numerical method for the system of HJB and transport equations}
is easily adapted to solve the fluid limit system
\eqref{fluid limit HJB non stationary}--\eqref{mfg transport equation fluid limit non stationary}.
As in Section~\ref{sec: Numerical method for HJB equation},
we employ method of lines to numerically solve the HJB equation.
The space derivative of $v_0(t, x)$ at each spatial grid point $x_m$ is approximated by a backward difference quotient
$\frac{\partial}{\partial x} v_0(t, x_m) \simeq \frac{v_0(t, x_m) - v_0(t, x_{m-1})}{\Delta x}$ so that  $\frac{\partial}{\partial t} v_0(t, x_m)$ becomes a function of  $v_0(t, x_m)$ and
 $v_0(t, x_{m-1})$:
\begin{multline}
 \frac{\partial}{\partial t} v_0 (t , x_m)    =  r v_0(t, x_m)
 - \frac{1}{2\beta_1} \left[ \left(p_0(t) - \kappa_1  - \frac{v_0(t, x_m) - v_0(t, x_{m-1})}{\Delta x} \right)^+ \right]^2 		\\
- \frac{1}{2\beta_2} \left[  \left( \lambda \delta \frac{v_0(t, x_m) - v_0(t, x_{m-1})}{\Delta x} - \kappa_2  \right)^+  \right]^2 ,
\quad m = 1, 2, \ldots, M.
\label{eqn: system of numerical equations fluid limit}
\end{multline}
We use Matlab's Runge-Kutta ODE solver \texttt{ode45} to solve the system
\eqref{eqn: system of numerical equations fluid limit}
of ordinary differential equations for $\{v(t, x_m): m = 0, 1, \ldots, M\}$ backward in time with boundary condition
$v_0(t, x_0) \equiv v_0(t, 0)$ given by \eqref{v0 boundary condition} and initial condition
$v(T, x_m)=0$ for all $m=0, 1, ..., M$.

We use forward in time and forward in space scheme to solve the transport equation
\eqref{mfg transport equation fluid limit non stationary}.
As in section~\ref{sec: Numerical method for transport equation}, we also prescribe the boundary condition $\eta_0(t_n, x_M) = 0$, $n=0, ..., N$ at $x_M \equiv X_{max}$ which assumes that $X_{max}$ is larger than the saturation level. 
We directly set $\eta_0(t_n, x_0) = 1$ and 
obtain the numerical values of
$\eta_0(t_{n+1}, x_m)$ for $m=1,\ldots, M$ via
\begin{align}
 \eta_0(t_{n+1}, x_m)  &=
\eta_0(t_{n}, x_m)
+\Delta t \left[ - \lambda \delta a_0(t_n, x_m)
 + q_0(t_n, x_m) \right]
\frac{\eta_0(t_n, x_{m+1}) - \eta_0(t_n, x_{m})}{\Delta x}.
\notag	
\end{align}

Figure~\ref{fig: fluid limit Q and R} illustrates the resulting solution both in the time-dependent model described above (left panel) and its stationary version (middle and right panels) similar to Section \ref{sec: Stationary mean field game Nash equilibrium}. We observe  two distinct features of interest. First, we find that uncertainty discourages exploration as the discounting effect lowers the NPV of putting in effort today for a delayed reward at discovery date $\tau$. As a result, more uncertainty decreases aggregate production $\tilde{Q}$ and raises prices. Second, uncertainty encourages ``hoarding'', i.e.~holding additional reserves as a buffer against running out due to depletion Consequently, $\tilde{R}_\epsilon$
increases in $\epsilon$ (right panel of Figure~\ref{fig: fluid limit Q and R}). At the same time, as $\epsilon \downarrow 0$ stationary reserves level $\tilde{R}_\epsilon \downarrow 0$. Indeed, in the limit $\epsilon=0$, production can be viewed as a perfect just-in-time supply chain: effort is expended to find an infinitesimal amount of new underground resources which are immediately extracted and sold for profit. Thus, exploration effort becomes equivalent to a secondary production cost, the cost of securing the commodity supply to exactly match the desired production rate, and the precautionary need for reserves vanishes. Thus we conclude that economically uncertainty regarding discoveries carries a \emph{cost}.

\begin{figure}[htb]
\begin{center} \hspace*{6pt}
\begin{tabular}{rcc}
\begin{minipage}{0.3\textwidth}
\includegraphics[width=0.98\textwidth,height=2in,trim=1.35in 2.95in 1.25in 3in]{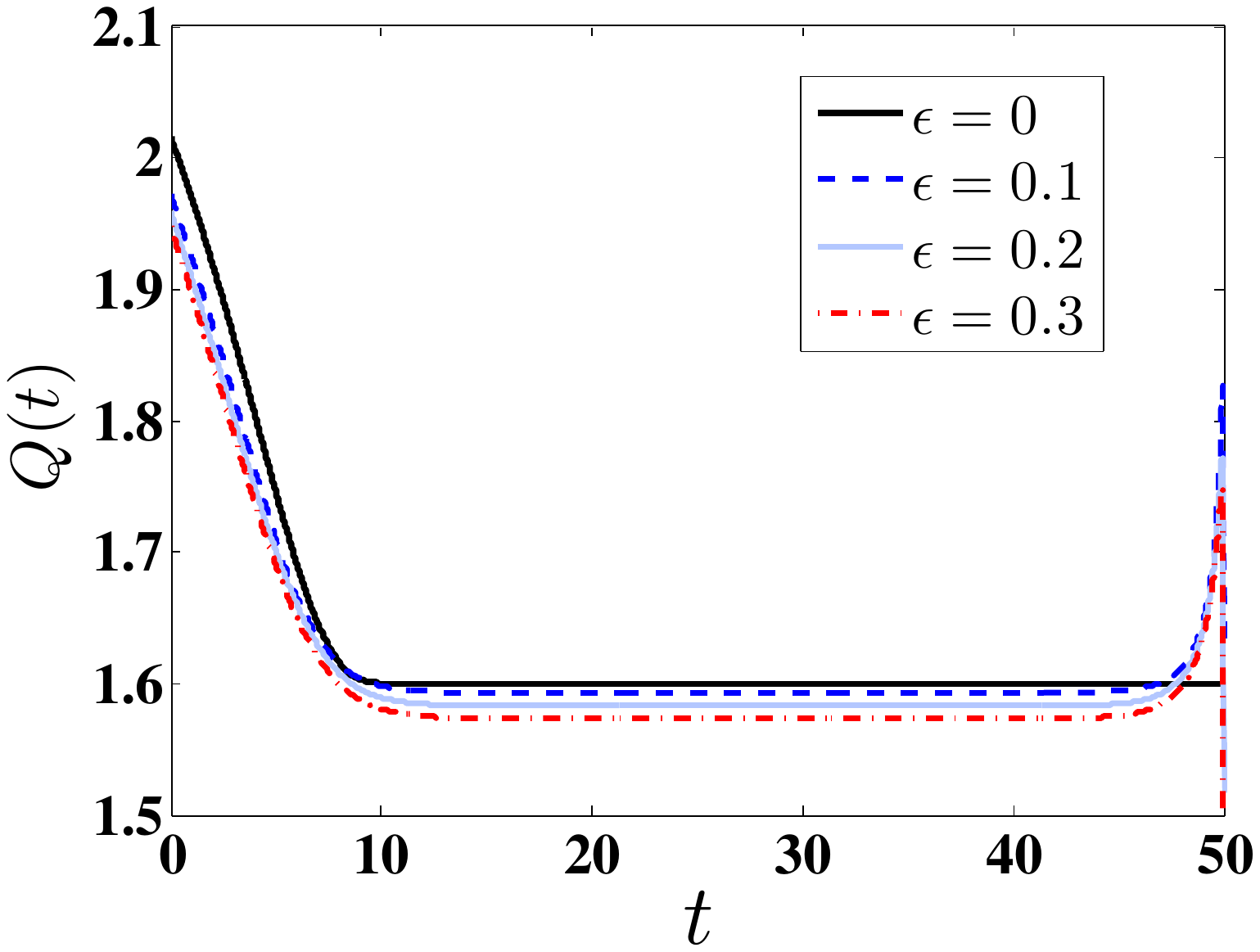}
\end{minipage} &
\begin{minipage}{0.3\textwidth}
\includegraphics[width=0.98\textwidth,height=2in,trim=1.45in 3.05in 1.45in 3.05in]{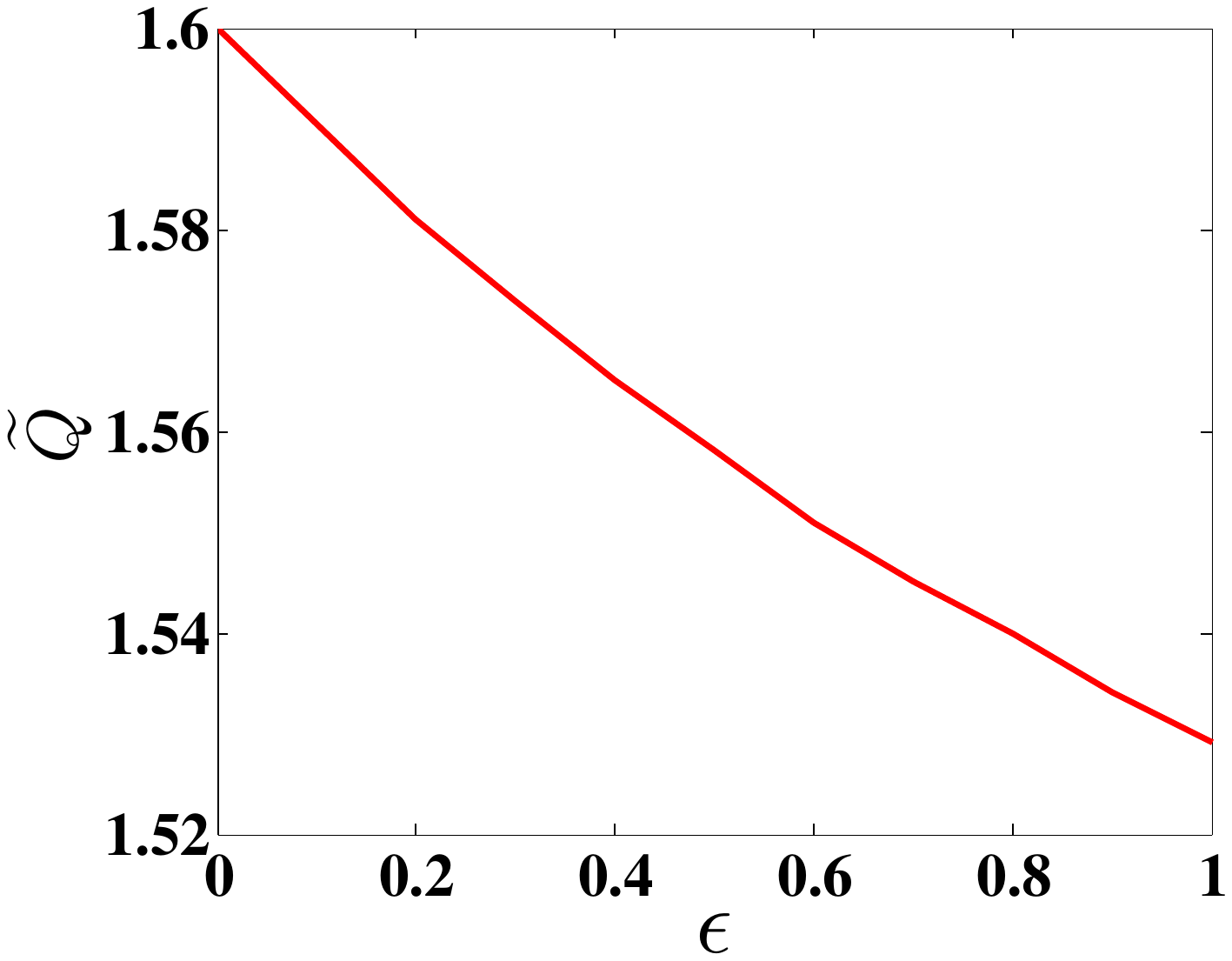}
\end{minipage} &
\begin{minipage}{0.3\textwidth}
\includegraphics[width=0.98\textwidth,height=2in,trim=1.35in 2.95in 1.35in 3in]{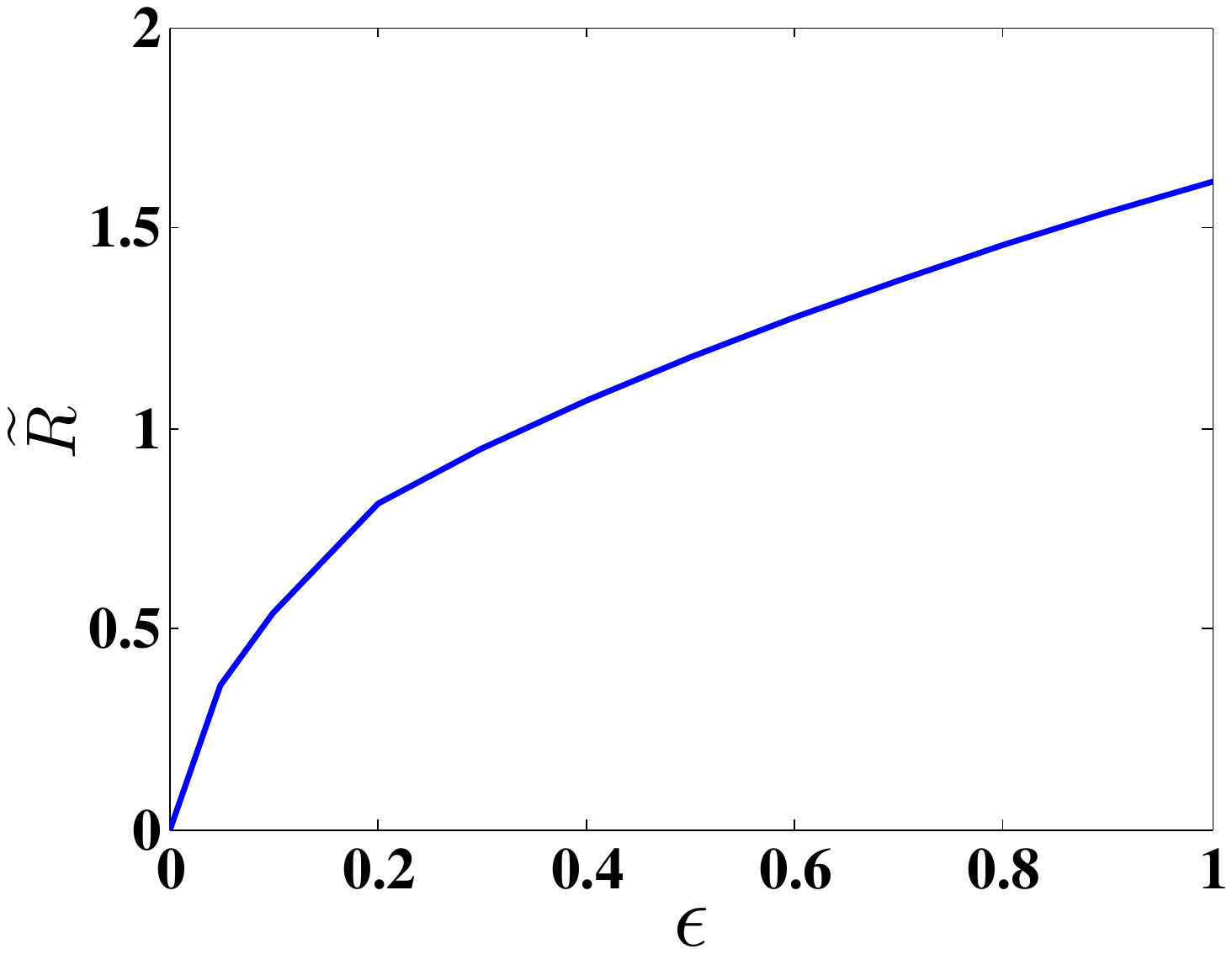}
\end{minipage}
\end{tabular}
\begin{minipage}{0.97\textwidth}
\caption{Equilibrium production and reserves level in the regime $\lambda_\epsilon = \lambda/ \epsilon$ and $\delta_\epsilon = \delta \epsilon$ for different values of $\epsilon$.
\emph{Left} panel: Evolution of total production $Q_\epsilon(t)$ for several levels of $\epsilon$.
\emph{Middle:} Stationary production $\tilde{Q}_\epsilon$ against $\epsilon$.
\emph{Right:} Stationary reserves level $\tilde{R}_\epsilon$ against $\epsilon$.
For $\epsilon =0$ we have $\tilde{Q}_0 = \frac{5 \cdot 1-0.2}{1 + 2\cdot 1} = 1.6$ from \eqref{total production in fluid limit} and $\tilde{R}_0 = 0$.
\label{fig: fluid limit Q and R}}
\end{minipage}\end{center}
\end{figure}

\section{Conclusion}\label{sec:conclusion}

We investigate joint production and exploration of exhaustible commodities in a mean-field oligopoly. The ability to expend effort to find new resources creates several new phenomena that modify both the mathematical and economic structure of the market. First, exploration weakens the exhaustibility constraint and in particular permits existence of a stationary model where individual producer reserves evolve, but the market price and aggregate quantities are invariant over time. Second, exploration modifies the role of holding reserves ---rather than determining future available means of production, reserves are partly used as a buffer to mitigate running out. As a result, if exploration is instantaneous and deterministic, no reserves are needed. This was explored in our analysis of the fluid limit model and connects to the early single-agent works from 1970s. Third, exploration control brings novel mathematical challenges to Cournot MFGs, in particular due to the non-local term (from discrete reserve events) in the transport equation and the non-smooth reserves distribution that involves a point mass at $x=0$ and a density on $(0,\infty)$. Fourth, the time-stationary Cournot MFG creates to a non-standard ``free boundary'' feature which required an approximation with a time-dependent version.

Among our insights is analysis about the ambiguous effects of exploration uncertainty and exploration frequency on the MFG equilibrium. This highlights the intricate interaction between stochasticity, reserves and the two types of controls, in addition to the game effects.
A further important contribution is development of tailored numerical schemes to solve the various versions of the Cournot models which require different handling of the boundary conditions, of space- and time-dimensions and of the first-order-condition terms that determine the optimal controls.

In our illustrations, the role of time horizon $T$ was mainly secondary and only affected the discovery rate $\lambda(t)$. A more extensive calibration could be made to add additional $t$-dependency, which could be used as means to capture learning-by-doing, or to capture the intuition that discovery sizes might get smaller over time. 

Another variant of the presented MFG approach would be to consider competition between a single major energy producer and a large population of minor energy producers, cf.~\cite{Huang10,CarmonaZhu16}.
This would correspond for example to the dominant role played by the Organization of Petroleum Exporting Countries (OPEC) in the crude oil market, with OPEC controlling about 40\% of the world's oil production. Due to the resulting market power, the
minor producers choose production strategies based on the production strategy of OPEC. The corresponding game model would involve a  game value function for the major player, a game value function for a representative minor producer, and the reserves distribution of minor producers. The price is then determined by the aggregate production of the major plus all minor producers.

Another open problem is to establish the existence and uniqueness of the MFG MNE with stochastic discoveries, and the regularity of the associated value function. As discussed, the corresponding reserves distribution is non-smooth with a point mass at $x=0$, so only weak regularity is expected. Intuitively, better regularity might be possible if the discovery distribution is continuous (rather than a fixed amount $\delta$), although this could generate further challenges for the HJB equation, introducing a bonified integral term into \eqref{mfg HJB equation}.  Such theoretical analysis could also help to rigorize the convergence of the proposed numerical scheme.


\appendix
\section{Appendix}

\subsection{Proof of Proposition \ref{prop: transport equation}}
\label{Proof of proposition transport equation}

\begin{proof}
Given $T$,
let $h(t,x) \in C^\infty_c((0, T) \times \mathds{R_+})$ be a test function that  is compactly  supported in
$(0, T) \times \mathds{R_+}$.
Using It\^{o}'s formula for jump processes, and $h(T, X_T) = h(0, X_0) = 0$ we have
\begin{align}
0
& = \E\left[ h(T,X_T) - h(0,X_0)   \right] 	\notag	\\
& = \E\left[ \int_0^T \frac{\partial}{\partial t}h(t,X_t) - q(t, X_t)\frac{\partial}{\partial y}h(t,X_t) dt
+ \int_0^T [h(t, X_t) - h(t-,X_{t-})] d N_t  \right] 	\notag	\\
& =-\int_{0}^\infty \int_0^T \frac{\partial}{\partial t}h(t,x) \frac{\partial}{\partial x}\eta(t,x) dt dx
+ \int_{0}^\infty \int_0^T q(t,x) \frac{\partial}{\partial x}h(t,x) \frac{\partial}{\partial x}\eta(t,x) dt dx  \notag	\\
& \quad
- \int_{\delta}^{\infty} \int_0^T \left( h(t, x) - h(t, x-\delta) \right)
\lambda(t) a(t,x-\delta) \frac{\partial}{\partial x}\eta(t,x-\delta) dt dx  \notag	\\
& \quad + \int_0^T h(t,\delta) \lambda(t) a(t,0)\pi(t) dt  =: I1 +I2 +I3 + I4.
\label{transport equation proof 1}
\end{align}

By integration-by-parts and the fact that $h(t,x)$ has compact support in $(0,T)\times \mathds{R_+}$, the first term on the right hand side of the last equality of
\eqref{transport equation proof 1} equals to
\begin{align}
\label{transport equation proof term 1}
I1 & = \int_{0}^\infty \int_0^T \eta(t,x) \frac{\partial}{\partial t} \frac{\partial}{\partial x} h(t,x)  d t  dx 	=  - \int_{0}^\infty \int_0^T
\frac{\partial}{\partial x} h(t,x)  \frac{\partial}{\partial t} \eta(t,x)  d t  dx .
\end{align}

By defining $F(t,x):=\int_0^x \lambda(t) a(t,z) \eta(t,dz)$, $x>0$,
the third term on the right hand side of the last equality of \eqref{transport equation proof 1} can be written as
\begin{align}
\label{transport equation proof term 3}
I3 &
=
 - \int_{\delta}^{\infty}  \int_0^T   h(t, x)
\frac{\partial}{\partial x} F(t, x-\delta)  dt dx
 + \int_{0}^{\infty}  \int_0^T   h(t, x)  \frac{\partial}{\partial x} F(t, x)   dt \, dx \notag	\\
&
=
  \int_{\delta}^{\infty}  \int_0^T
 F(t, x-\delta)  \frac{\partial}{\partial x} h(t, x)   dt dx
-  \int_{0}^{\infty}  \int_0^T   F(t, x)  \frac{\partial}{\partial x} h(t, x)   dt   \notag	\\
&
=
  \int_{\delta}^{\infty}  \int_0^T
 \left( F(t,x) - F(t, x-\delta)  \right)  \frac{\partial}{\partial x} h(t, x)   dt dx
 - \int_{0}^{\delta}  \int_0^T   F(t, x)  \frac{\partial}{\partial x} h(t, x)   dt \, dx.
\end{align}

The fourth term on the right hand side of the last equality of \eqref{transport equation proof 1} can be written as
\begin{align}
\label{transport equation proof term 4}
I4 & =
 \int_0^T \left( \int_0^\delta \frac{\partial}{\partial x}h(t,x) dx \right)
\lambda(t) a(t,0)\pi(t) dt	\notag	\\
& =
\int_0^\delta \int_0^T \lambda(t) a(t,0)\pi(t) \frac{\partial}{\partial x}h(t,x)
 dt dx .
\end{align}

By substituting \eqref{transport equation proof term 1}-
\eqref{transport equation proof term 4} into equation \eqref{transport equation proof 1},
we obtain
\begin{align}
0 & =
 - \int_{0}^\delta \int_0^T
\frac{\partial}{\partial x} h(t,x)
\left[ \frac{\partial}{\partial t} \eta(t,x)
-  q(t,x) \frac{\partial}{\partial x}\eta(t,x)
+ \int_{0+}^x\lambda a(t,z) \eta(t,dz)	\right. 	\notag	\\
& \quad\qquad\qquad\qquad\qquad
 + \lambda(t) a(t,0) \pi(t) \Bigr] dt \, dx   	\notag	\\
& 	\quad
 - \int_{\delta}^\infty \int_0^T
\frac{\partial}{\partial x} h(t,x)
\left[ \frac{\partial}{\partial t} \eta(t,x)
-  q(t,x) \frac{\partial}{\partial x}\eta(t,x)
+ \int_{x-\delta}^x\lambda a(t,z) \eta(t,dz)
\right]
dt \, dx,
\label{transport equation proof 2}
\end{align}
which is true for any test function $h(t,x) \in C^\infty_c\left( (0,T)\times\mathds{R_+} \right)$.
According to the first term of the right hand side of \eqref{transport equation proof 2}, we have
\begin{align}
0=
\frac{\partial}{\partial t} \eta(t,x)
-  q(t,x) \frac{\partial}{\partial x}\eta(t,x)
+ \int_{0+}^x\lambda(t) a(t,z) \eta(t, d z)+ \lambda(t) a(t,0) \pi(t), \quad
0 < x < \delta.
\label{transport equation proof piece 2}
\end{align}
According to the second term of the right hand side of \eqref{transport equation proof 2}, we have
\begin{align}
\label{transport equation proof piece 3}
0=
\frac{\partial}{\partial t} \eta(t,x)
-  q(t,x) \frac{\partial}{\partial x}\eta(t,x)
+ \int_{x-\delta}^x\lambda(t) a(t,z) \eta(t,dz) , \quad
x > \delta.
\end{align}

The two equations
\eqref{transport equation proof piece 2} - \eqref{transport equation proof piece 3}
constitute the transport equation of reserves distribution given in Proposition
\ref{prop: transport equation}.

\end{proof}

\subsection{Proof of Lemma \ref{lemma: relation of Q A R}}
\label{app: relation of Q A R}

\begin{proof}
We integrate both sides of the transport equation \eqref{mfg transport equation} with respect to $x$ over $(0, \infty]$ to obtain
\begin{align}
\int_{0+}^\infty \frac{\partial}{\partial t}\eta^\ast(t, x) dx
& =
- \overbrace{\int_{0+}^\delta  \left( \int_{0+}^x \lambda(t) a^\ast(t,z) \eta^\ast(t, dz) \right)  \ dx}^{ =: I1} 	
 - \overbrace{\int_\delta^\infty \left( \int_{x-\delta}^x \lambda(t) a^\ast(t,z) \eta^\ast(t,dz) \right)  \ dx}^{ =:I2}   \notag \\
& \quad +  \int_{0+}^\infty \left( q^\ast(t,x)\frac{\partial}{\partial x} \eta^\ast(t,x) \right)  dx	.
\label{relation of Q A R equation 1}
\end{align}
For the last term by definition of the Stieltjes integral, the integrator is equivalently $\frac{\partial}{\partial x} \eta^\ast(t,x)  dx = \eta^\ast(t,dx)$.
We apply integration by parts to the first two terms on the RHS of
\eqref{relation of Q A R equation 1} to obtain
\begin{align}
I1 & = \left[ x\int_{0+}^x \lambda(t) a^\ast(t,z) \eta^\ast(t, dz) \right]_{0+}^\delta 	
 - \int_{0+}^\delta x \frac{\partial}{\partial x}
\left( \int_{0+}^x \lambda(t) a^\ast(t,z) \eta^\ast(t, dz) \right) dx	\notag	\\
& = \delta \int_{0+}^\delta \lambda(t) a^\ast(t,z) \eta^\ast(t, dz)
- \int_{0+}^\delta x \lambda(t) a^\ast(t,x) \eta^\ast(t, dx) , \qquad \text{and}
\label{relation of Q A R equation 2}
\end{align}
\begin{align}
I2 &
=	\left[ x  \int_{x-\delta}^x \lambda(t) a^\ast(t,z) \eta^\ast(t,dz) \right]_{x=\delta}^{x=\infty}
- \int_\delta^\infty x \frac{\partial}{\partial x}\left( \int_{x-\delta}^x \lambda(t) a^\ast(t,z) \eta^\ast(t,dz) \right) dx
\notag	\\
&
= - \delta \int_{0+}^\delta \lambda(t) a^\ast(t,z) \eta^\ast(t,dz)
-\left[ \int_\delta^\infty x \lambda(t) a^\ast(t,x) \eta^\ast(t, dx)
- \int_{0+}^\infty (x+\delta) \lambda(t) a^\ast(t,x) \eta^\ast(t, dx) \right] \notag	\\
& =
- \delta \int_{0+}^\delta \lambda(t) a^\ast(t,z) \eta^\ast(t,dz)
+ \int_{0+}^\delta x \lambda(t) a^\ast(t,x) \eta^\ast(t, dx)
+ \delta \int_{0+}^\infty \lambda(t) a^\ast(t,x) \eta^\ast(t, dx) .
\label{relation of Q A R equation 3}
\end{align}
The left hand side of \eqref{relation of Q A R equation 1} can be written as
\begin{align}
\int_{0+}^\infty \frac{\partial}{\partial t}\eta^\ast(t,x)dx
 =  \frac{d}{d t} \int_{0+}^\infty  \eta^\ast(t,x)dx ,
\label{relation of Q A R equation 4}
\end{align}
where the exchange of the partial differential operator and the integral is justified by the Leibniz integral rule under the condition that both $\eta^\ast(t,x)$ and $\frac{\partial}{\partial t}\eta^\ast(t,x)$ are continuous in the domain $(t,x) \in [0, \infty)\times(0, \infty)$.
By substituting \eqref{relation of Q A R equation 2}--\eqref{relation of Q A R equation 4} into the equation \eqref{relation of Q A R equation 1}, we have
\begin{align}
 \frac{d}{d t} \int_{0+}^\infty  \eta^\ast(t,x)dx
=- \delta \int_{0+}^\infty \lambda(t) a^\ast(t,x) \eta^\ast(t, dx)
+ \int_{0+}^\infty q^\ast(t,x) \eta^\ast(t,dx) ,
\notag
\end{align}
which gives \eqref{relation of Q A R differential form}.
\end{proof}

\subsection{Proof of Proposition \ref{prop: stationary mean field game equilibrium in fluid limit}}
\label{section: proof of stationary mean field game equilibrium in fluid limit}

We first present Lemmas
\ref{lemma fluid limit HJB transport equations}
and \ref{lemma: equilibrium production and exploration in fluid limit} that summarize
the partial differential equations associated with the fluid limit of our MFG model.
\begin{lemma}
\label{lemma fluid limit HJB transport equations}
The limiting game value function $v_0$ and reserves distribution function
$(\tilde{\pi}_0,  \tilde{\eta}_0)$ satisfy the following system
\begin{align}
\label{mfg HJB fluid limit}
 r \tilde{v}_0(x) & =
 \left[ \left( \tilde{p}_0
- \tilde{v}'_0(x)   \right)  \tilde{q}^\ast_0(x) - C_q( \tilde{q}^\ast_0(x))  \right]  + \left[ - C_a( \tilde{a}^\ast_0(x) ) + \tilde{a}^\ast_0(x)\lambda \delta
\tilde{v}'_0(x) \right]  , \quad
 x \ge 0 ,
\end{align}
\begin{align}
\label{mfg transport fluid limit}
\begin{cases}
 0 = -\lambda\delta \tilde{a}^\ast_0(0) \tilde{\pi}_0
- \tilde{q}^\ast_0(0) \tilde{\eta}'_0(0+) , \\
 0 =
\left(- \lambda \delta \tilde{a}_0^\ast(x) + \tilde{q}_0^\ast(x) \right) \tilde{\eta}'_0(x), \quad x > 0 ,
\end{cases}
\end{align}
where the optimal production rate $\tilde{q}^\ast_0$ and exploration rate $\tilde{a}^\ast_0$ are given by
\begin{align}
& \tilde{q}^\ast_0(x)
= \frac{1}{\beta_1} \left( L  -  \tilde{Q}_0 - \kappa_1
-  \tilde{v}_0'(x) \right)^+ ,  	\notag	\\
& \tilde{a}^\ast_0(x)
 =  \frac{1}{\beta_2}\left( \lambda \delta  \tilde{v}'_0(x) - \kappa_2 \right)^+  ,
\label{production and exploration fluid limit}
\end{align}
and aggregate production $\tilde{Q}_0$ is uniquely determined by the equation
\begin{align}
\tilde{Q}_0
& =
- \int_0^\infty
\frac{1}{\beta_1}\left( L - \kappa_1 -  \tilde{v}_0'(x) - \tilde{Q}_0 \right)^+
\tilde{\eta}_0(dx) ,   \label{eq:tildeQ0}
\end{align}
and the equilibrium price is
\begin{align}
\tilde{p}_0 = L + \int_0^\infty \tilde{q}^\ast_0(z) \tilde{\eta}_0(dz) .
\end{align}

\end{lemma}

\begin{proof}

To obtain the HJB equation \eqref{mfg HJB fluid limit} of limiting game value function $\tilde{v}_0(x)$ and the associated optimal production controls
\eqref{production and exploration fluid limit},
we let $\epsilon \to 0$, so that $\delta \epsilon \downarrow 0$ and $\lambda/\epsilon \uparrow \infty$ in the HJB equation
\eqref{mfg HJB} and the associated optimal production and exploration controls
\eqref{eq:q-star}-\eqref{optimal mfg q and a explicit}
\begin{align}
\tilde{a}^\ast_0(x)  = \lim_{\epsilon \to 0} \tilde{a}^\ast_\epsilon (x)
&  =  \lim_{\epsilon \to 0} \frac{1}{\beta_2}
\left[ \lambda_\epsilon \left(\tilde{v}_\epsilon(x+\delta_\epsilon) - \tilde{v}_\epsilon(x) \right) - \kappa_2 \right]^+  \notag	\\
&  =  \lim_{\epsilon \to 0} \frac{1}{\beta_2}
\left[  \frac{\lambda}{\epsilon} \left( \tilde{v}_\epsilon(x+\delta \epsilon) - \tilde{v}_\epsilon(x) \right) - \kappa_2 \right]^+ 	\notag	\\
& =   \frac{1}{\beta_2}
\left(  \lambda \delta \tilde{v}_0'(x)  - \kappa_2 \right)^+ ,
\notag \\
\tilde{q}_0^\ast(x) =  \lim_{\epsilon \to 0} \tilde{q}^\ast_\epsilon (x)
& = \lim_{\epsilon \to 0} \frac{1}{\beta_1} \left( L  -  \tilde{Q}_\epsilon - \kappa_1
- \tilde{v}_\epsilon'(x) \right)^+ \notag \\
& = \frac{1}{\beta_1} \left( L  -  \tilde{Q}_0 - \kappa_1
-  \tilde{v}_0'(x) \right)^+.
\notag
\end{align}
Equations \eqref{mfg transport fluid limit} and \eqref{eq:tildeQ0} follow similarly from \eqref{stationary mfg transport equation} and  \eqref{eq:stat-Q}.
Note that as $\epsilon \downarrow 0$ the integral term
$\int_{x - \delta_\epsilon}^x \lambda_\epsilon  \tilde{a}_\epsilon(z)  \tilde{\eta}_\epsilon(dz)$
in the third case $\delta_\epsilon < x$ of \eqref{eq:eta-2}
converges to
\begin{align}
\lim_{\epsilon \to 0} \int_{x - \delta_\epsilon}^x \lambda_\epsilon \tilde{a}_\epsilon(z) \tilde{\eta}_\epsilon(dz)
=
\lim_{\epsilon \to 0} \int_{x - \delta \epsilon}^x  \frac{\lambda}{\epsilon} \tilde{a}_\epsilon( z) \tilde{\eta}_\epsilon( dz)
=
\lambda \tilde{a}_0( x) \frac{\partial}{\partial x} \tilde{\eta}_0( x) .
\notag
\end{align}
\end{proof}

\begin{lemma}[Fluid limit stationary boundary condition at $x=0$]
\label{lemma: equilibrium production and exploration in fluid limit}
The equilibrium production and exploration rates in fluid limit on the boundary $x=0$ satisfy
\eqref{equation: fluid limit boundary production and exploration relation}.

\end{lemma}

\begin{proof}
On the boundary $x=0$ where there is no reserves, we must have
$\delta \lambda \tilde{a}^\ast_0(0) \geq \tilde{q}^\ast_0(0) \geq 0$, i.e., the rate of reserves addition must be greater than or equal to production rate.
If $a^\ast_0(0) = 0$, it follows that
$q^\ast_0(0) = \lambda \delta a^\ast_0(0) = 0$. Now we consider the case that $a^\ast_0(0) > 0$.
Since $\tilde{q}^\ast_0(x) \geq 0$ is non-decreasing, and $\tilde{a}^\ast_0(x)$ decreases to $0$
as $x$ increases, we must have some point $x^\ast \geq 0$ such that
$q_0^\ast(x^\ast) = \lambda \delta a_0^\ast(x^\ast)$.
Note that once reserves process $X_t$ reaches the level $x^\ast$, it will remain unchanged since production rate $\tilde{q}_0^\ast(x^\ast)$ is balanced by the rate of reserves increment
$\lambda \delta \tilde{a}_0^\ast(x^\ast)$ at $X_t = x^\ast$. We now prove that $x^\ast = 0$. Towards a contradiction suppose that $x^\ast >0$.
Then  
\begin{align}
\tilde{v}_0(x^\ast)
& =
\int_0^{\infty}\left[ \tilde{p}_0 \tilde{q}^\ast_0(X^{ x^\ast }_t) - C_q(\tilde{q}^\ast_0(X^{ x^\ast }_t)) - C_a(\tilde{a}^\ast_0(X^{ x^\ast }_t)) \right] e^{-rt} dt  \notag \\
& =
\int_0^{\infty}\left[ \tilde{p}_0 \tilde{q}^\ast_0(x^\ast) - C_q(\tilde{q}^\ast_0(x^\ast)) - C_a(\tilde{a}^\ast_0(x^\ast)) \right] e^{-rt} dt  
 \leq \tilde{v}_0(0) , \notag
\end{align}
since starting at $x^*$, $X^{x^*}_t = x^*$ for all $t$, so that the resulting strategy $(\tilde{q}^\ast_0(x^\ast),\tilde{a}^\ast_0(x^\ast))$ is admissible for zero initial reserves and hence sub-optimal for the problem defining $\tilde{v}_0(0)$.

Next, let  $\tau:= \inf \left\{ t\geq 0: \int_0^t q^\ast_0(0) ds = x^\ast   \right\}$. 
\begin{align}
\tilde{v}_0(0)
& = \int_0^\tau  \left[ \tilde{p}_0 q^\ast_0(X^0_t) - C_q(q^\ast_0(X^0_t))-C_a(\tilde{a}^\ast_0(X^{0}_t)) \right] \e^{-rt} dt
+ \e^{-r \tau} \tilde{v}_0(X^0_\tau)   \notag \\
& < \int_0^\tau  \left[ \tilde{p}_0 q^\ast_0(X^0_t) - C_q(q^\ast_0(X^0_t)) \right] \e^{-rt} dt
+ \e^{-r \tau} \tilde{v}_0(X^0_\tau) \le \tilde{v}_0(x^\ast)  
\end{align}
where the strict inequality $``>"$ is due to the assumption that $a^\ast_0(0) > 0$ and the last inequality is due to the resulting strategy of starting at $x^\ast$, running reserves down to zero with zero exploration and then continuing with $\tilde{v}_0(X^{x^\ast}_\tau)$ is admissible for the problem defining $\tilde{v}_0(x^\ast)$ (treating $\tilde{p}_0$ fixed).  The above two inequalities contradict each other;
thus $x^\ast = 0$, and  $\tilde{q}_0^\ast(0) = \lambda \delta \tilde{a}_0^\ast(0)$. \end{proof}

 Using Lemmas \ref{lemma fluid limit HJB transport equations} and \ref{lemma: equilibrium production and exploration in fluid limit},
we  prove Proposition \ref{prop: stationary mean field game equilibrium in fluid limit}.

\begin{proof}[Proof of proposition \ref{prop: stationary mean field game equilibrium in fluid limit}]
\emph{(i)}.
Since $\tilde{q}^\ast_0(x) > 0$ and $\tilde{a}^\ast_0(x) = 0$ for $x>0$,
we have $\tilde{\eta}'_0(x) = 0$ for $x>0$,
according to the equation
\eqref{mfg transport fluid limit} in the interior. It follows that the reserves distribution degenerates to a point mass at 0, $\tilde{\pi}_0 = 1$. Substituting the latter fact into \eqref{eq:tildeQ0} we get
\begin{align}
\tilde{Q}_0
& =  \frac{1}{\beta_1} \left( L - \kappa_1 - \tilde{v}_0'(0) - \tilde{Q}_0 \right)^+ ,
\notag
\end{align}
which gives
$\tilde{Q}_0 =  \frac{1}{1 + \beta_1} \left( L - \kappa_1 -  \tilde{v}_0'(0) \right)^+$ .
According to \eqref{production and exploration fluid limit}
the equilibrium production rate at $x=0$ is
\begin{align}
\tilde{q}^\ast_0(0)
& = \frac{1}{\beta_1}\left( L - \kappa_1 - \tilde{v}_0'(0) - \tilde{Q}_0 \right)^+ \notag \\
& = \frac{1}{1 + \beta_1} \left( L - \kappa_1 -  \tilde{v}_0'(0) \right)^+  .
\label{stationary equilibrium production 1 fluid limit}
\end{align}

By substituting production rate \eqref{stationary equilibrium production 1 fluid limit}
and  exploration effort
\eqref{production and exploration fluid limit} with $x=0$
into the equation $\tilde{q}^\ast_0(0) = \lambda\delta \tilde{a}^\ast_0(0)$
and solving for $\tilde{v}_0'(0)$, we obtain
\begin{align}
\tilde{v}_0'(0) = \frac{(L - \kappa_1)\beta_2 + \kappa_1(1+\beta_1)}{\beta_2 + (1+\beta_1)\lambda\delta} .
\label{v prime zero fluid limit}
\end{align}
Then by substituting the above $\tilde{v}_0'(0)$ into
\eqref{stationary equilibrium production 1 fluid limit} we have
\begin{align}
\tilde{q}_0^\ast(0)
= \frac{ [ (L-\kappa_1)\lambda \delta - \kappa_2 ]^+ }{\beta_2 + (1+\beta_1)\lambda \delta} .
\notag
\end{align}
Since $\tilde{\pi}_0 = 1$, we have $\tilde{Q}_0 = - \int_0^\infty \tilde{q}^\ast_0(x) \tilde{\eta}_0(dx)
= \tilde{q}^\ast_0(0)$ which yields \eqref{total production in fluid limit}. \end{proof}

\subsection{Proof of Lemma \ref{lemma: boundary conditions fluid limit non stationary HJB}}
\label{proof of lemma boundary conditions fluid limit non stationary HJB}

\begin{proof}
Similar to Lemma \ref{lemma: equilibrium production and exploration in fluid limit}, at $x=0$ we have
\begin{align}
q_0^\ast(t, 0) = \lambda \delta a_0^\ast(t, 0) .
\label{production and exploration on boundary fluid limit non stationary}
\end{align}
By substituting
\eqref{optimal production fluid limit non stationary}-\eqref{optimal exploration fluid limit non stationary} into
\eqref{production and exploration on boundary fluid limit non stationary},
we obtain the boundary condition 
\begin{align}
\frac{\partial}{\partial x} v_0(t, 0)
= \frac{\beta_2 (p_0(t) - \kappa_1 ) + \beta_1 \lambda \delta \kappa_2}{\beta_1 \lambda^2 \delta^2 + \beta_2} .
\notag
\end{align}
Substituting \eqref{v0 partial x boundary condition} into
\eqref{optimal production fluid limit non stationary}-\eqref{optimal exploration fluid limit non stationary}, we obtain $a^\ast_0(t, 0)$ and $q^\ast_0(t, 0)$ in explicit form
\begin{align}
& a^\ast_0(t, 0) = \frac{\lambda \delta (p_0(t) - \kappa_1) - \kappa_2}{\beta_1 \lambda^2 \delta^2 + \beta_2},
\label{exploration on boundary fluid limit non stationary}
\\
& q^\ast_0(t, 0) =
\frac{\lambda^2 \delta^2 (p_0(t) - \kappa_1 ) - \lambda \delta \kappa_2}{\beta_1 \lambda^2 \delta^2 + \beta_2} .
\label{production on boundary fluid limit non stationary}
\end{align}

By substituting \eqref{v0 partial x boundary condition}, \eqref{exploration on boundary fluid limit non stationary}, and
\eqref{production on boundary fluid limit non stationary}
into the HJB equation
\eqref{fluid limit HJB non stationary},
we obtain the following linear first-order differential equation for $v_0(\cdot, 0)$:
\begin{align}
0 &= \frac{\partial}{\partial t} v_0(t, 0)  - r v_0(t, 0)
 +  \frac{1}{2}\left[ ( a^\ast_0(t, 0) )^2 +  ( q^\ast_0(t, 0))^2   \right]  ,   \qquad
0< x , \ 0\leq t < T ,
\notag
\end{align}
which admits an explicit solution
\begin{align}
v_0(t, 0)
& = v_0(T, 0) \e^{-r(T-t)} + \int_t^T  \frac{1}{2}\left[ ( a^\ast_0(s, 0) )^2
+  ( q^\ast_0(s, 0))^2   \right] \e^{-r(s - t)} ds  \notag
\end{align}
that matches \eqref{v0 boundary condition} since $v_0(T,0) = 0$.
\end{proof}

\bibliographystyle{siam}
\bibliography{../../../../../masterBib}

\end{document}